\documentclass[letterpaper]{article} 
\usepackage{aaai20}  
\usepackage{times}  
\usepackage{helvet} 
\usepackage{courier}  
\usepackage[hyphens]{url}  
\usepackage{microtype}
\usepackage{subfigure}
\usepackage{booktabs} 
\usepackage{algorithm, algorithmic}

\usepackage{epsfig,amsmath,amssymb, mathrsfs,psfrag}
\usepackage{amsthm}
\usepackage{graphicx}
\usepackage{bbold}
\usepackage{bm}
\usepackage{caption}
\newcommand{\flip}{\text{flip}}
\usepackage[export]{adjustbox}
\newtheorem{theorem}{Theorem}[section]
\usepackage{subfigure}

\usepackage{hyperref} 
\hypersetup{
     colorlinks   = true,
     citecolor    = blue,
     linkcolor    = magenta, 
     pagebackref = true
}

\title{Biologically Plausible Sequence Learning with Spiking Neural Networks}

\author{
Zuozhu Liu,\textsuperscript{\rm 1,3}
Thiparat Chotibut,\textsuperscript{\rm 2,3,}\thanks{Correspondence to: thiparatc@gmail.com}
Christopher Hillar,\textsuperscript{\rm 4,5}
Shaowei Lin,\textsuperscript{\rm 3}\\
\textsuperscript{\rm 1}Department of Statistics and Applied Probability, National University of Singapore,\\
\textsuperscript{\rm 2}Department of Physics, Faculty of Science, Chulalongkorn University, Thailand,\\
\textsuperscript{\rm 3}Engineering Systems and Design, Singapore University of Technology and Design,
\textsuperscript{\rm 4}Awecom, Inc,\\
\textsuperscript{\rm 5}Redwood Center for Theoretical Neuroscience, University of California, Berkeley, \\
lcowen.hn@gmail.com, {\rm *}thiparatc@gmail.com, hillarmath@gmail.com, 
shaowei\_lin@sutd.edu.sg
}

\begin{document}
\maketitle
\begin{abstract}
Motivated by the celebrated discrete-time model of nervous activity outlined by McCulloch and Pitts in 1943, we propose a novel continuous-time model, the McCulloch-Pitts network (MPN), for sequence learning in spiking neural networks. Our model has a local learning rule, such that the synaptic weight updates depend only on the information directly accessible by the synapse. By exploiting {\it asymmetry} in the connections between binary neurons, we show that MPN can be trained to robustly memorize multiple {\it spatiotemporal} patterns of binary vectors, generalizing the ability of the symmetric Hopfield network to memorize static spatial patterns. In addition, we demonstrate that the model can efficiently learn sequences of binary pictures as well as generative models for experimental neural spike-train data. Our learning rule is consistent with spike-timing-dependent plasticity (STDP), thus providing a theoretical ground for the systematic design of biologically inspired networks with large and robust long-range sequence storage capacity.	
\end{abstract}

\section{Introduction}
\label{intro}

Experimental evidence from neurobiology reveals that changes in synaptic weights of some neurons depend on the timing difference between presynaptic and postsynaptic spikes  \cite{markram1995action,gerstner1996neuronal,bi2001synaptic}, a concept termed {\it spike-timing-dependent plasticity} (STDP). While deep learning is tremendously successful in numerous machine learning tasks, its underlying backpropagation learning algorithm is biologically implausible \cite{bengio2015towards}. Recent works attempt to bridge this gap by proposing STDP-consistent training rules for deep learning \cite{bengio2015towards,bengio2017stdp,liu2018variational,scellier2017equilibrium}, with the hope of uncovering computationally efficient learning algorithms inspired by biological brains.  Thorough understanding of STDP-consistent learning rules and their computing capabilities can also benefit the design of modern bioinspired computing hardware such as neuromorphic chips. 

On the other hand, machine learning can assist in the quest for understanding how brains perform computation, especially when  viewed through the lens of a spiking neural network (SNN). For example, variational inference and reinforcement learning unveil SNN architectures that perform probabilistic inference consistent with STDP rules \cite{nessler2013bayesian,pecevski2016learning,rezende2011variational}.
Learning-to-learn methods and Long Short-Term Memory help identify a novel SNN architecture with the capability to store long spatiotemporal sequences of spikes during computation \cite{bellec2018long}.
Moreover, building generative models of SNNs from experimental time-series data of action potentials yields insights into the network structure generating such spikes \cite{tyrcha2013effect,nasser2013spatio,zeng2013maximum}.

With the goal of building a biologically plausible generative model of SNNs from spatiotemporal patterns, such as neural spike-train data, we propose a novel continuous-time learning algorithm that is consistent with experimentally observed STDP.
We define simple rules for updating states and learning weights of a continuous-time SNN and demonstrate that the expected synaptic changes given a pair of presynaptic and postsynaptic spikes reproduce the STDP curve. 
As a consequence of our model, we show that the stochasticity of the refractory period is an important ingredient for reproducing STDP. Namely, by assuming that the time between the spiking of a neuron and its recovery is stochastic and follows an exponential distribution, we prove that the expected increase or decrease in synaptic weights follows STDP-like curves.

In addition to its biological plausibility, our SNN model is capable of learning sequences of binary vectors, or spatiotemporal patterns of spikes. In the deterministic limit, this network, a McCulloch-Pitts network \cite{mcculloch1943logical}, is a generalization of the Hopfield network that allows {\it asymmetric} connections between fully-visible binary units. 
As a result, deterministic dynamics in the state space can converge to an attractor of period larger than one (a cycle), as opposed to a fixed-point corresponding to a static spatial memory in the Hopfield network. 
We demonstrate that the network can be trained to memorize repeating sequences of binary vectors and that the spatiotemporal memory is robust to perturbation. In particular, given an arbitrary initial state of the spike pattern, the dynamics generated by the trained network converge to a memorized cycle. 
Thus, the STDP-consistent McCulloch-Pitts network performs a spatiotemporal pattern completion, generalizing the classical Hopfield network that performs spatial pattern completion. 
In the noisy limit, the trained network can be regarded as a generative model for spatiotemporal patterns of spikes that performs sampling of both the time to the next spiking event and the neuron to next spike. 
We illustrate that this network efficiently performs inference on continuous-time data from a neurobiology experiment, and we discuss its applicability to inference problems in neuroscience.

\section{McCulloch-Pitts Networks}
\label{sec:MPNetsModel}

To begin, suppose we are presented with a time series $\mathcal{D}(t)$ of binary vectors with single bit-flip transitions:
\begin{align*}
\begin{array}{rll}
\mathcal{D}(t) &= x^{(0)}\quad &\text{for }0= t_0 \leq t < t_1,\\
&\vdots & \\
\mathcal{D}(t) &= x^{(N-1)}\quad &\text{for } t_{N-1} \leq t < t_{N}, \\
\mathcal{D}(t) &= x^{(N)}\quad &\text{for } t= t_{N}.
\end{array}
\end{align*}
Our goal is to train a continuous-time stochastic process that mimics $\mathcal{D}(t)$.

In this section, we define our model of a neural stochastic process, which we term a McCulloch-Pitts network. Relationships to biological neural networks will be discussed in the Spike-Timing-Dependent Plasticity section. 
 Let $G=(V,E)$ be a weighted directed graph with vertices $V=\{1, \ldots, d\}$, edges $ij \in E$ from $i$ to $j$, vertex biases $b_1, \ldots, b_d \in \mathbb{R}$, and weights $w_{ij} \in \mathbb{R}$ for each edge $ij \in E$. Let $\theta =(b_i, w_{ij})$ denote the vector of parameters. Self-loops and directed cycles are allowed in $G$.

\subsection{Dynamics}
\label{sec:dynamics}

Let $\mathcal{X} = \{0,1\}^d$ be the state space of binary vectors on $V$, and let $X(t)$ be a stochastic process indexed by continuous time $t \geq 0$ with states $\mathcal{X}$. Let $\tau > 0$ be the \emph{temperature}\footnote{Temperature can be interpreted as a proper physical temperature in the Boltzmann machine since its symmetric weights uniquely determine states' energy function, and its thermodynamics formulation can be adopted. Here, in {\it asymmetric} networks, temperature reflects the frequency in which stochastic updates take place ({\it intrinsic noise}). Zero temperature thus corresponds to deterministic dynamics, whereas high temperature drives rapidly fluctuating stochastic dynamics.}, and define $\flip_i(x)$ to be the binary vector $x$ whose $i$-th coordinate is flipped. 
To generate $X(t)$, we select an initial condition $X(0)$ from $\mathcal{X}$, and repeat the following two steps:

\begin{enumerate}
\item Suppose $X(T)=x$. For each $i \in V$, we sample a \emph{holding time} $\xi_i$ from the exponential distribution $\text{Exp}(\lambda_i)$ where $\lambda_i = \exp( \sigma_i z_i /\tau)$, $\sigma_i = 1-2x_i$, and $z_i =\sum_{ji\in E} w_{ji}x_j + b_i$.
\item Let $\xi_j$ be the smallest of the holding times $\xi_i$, and let $T' = T+\xi_j$. We set $X(t) =x$ for all $T\leq t < T'$, and $X(T') = \flip_j(x)$. Finally, we update $T \leftarrow T'$.
\end{enumerate}

More generally, a rate scaling parameter $r_i$ could be introduced in the rate $\lambda_i = r_i \exp(\sigma_i z_i/\tau)$. However, $r_i$ can be absorbed into $b_i$ and $w_{ii}$, see Supplementary Materials (SM) A. Thus, we can set $r_i = 1$ for simplicity and assume the graph $G$ has self-loops. We term $X(t)$ the \emph{McCulloch-Pitts network} (MPN) associated to $(G,\theta,\tau)$. Observe that the transitions in $X(t)$ involve flipping only one bit at a time.

Alternatively, we can reformulate an equivalent yet simpler algorithm to generate $X(t)$ using softmax; given an initial condition $X(0) \in \mathcal{X},$ we repeat the following steps:
\begin{enumerate}
\item Let $X(T)=x$. Sample $\xi \sim \text{Exp}(\lambda)$ where $\lambda = \sum_{i\in V} \lambda_i$.
\item Let $T' = T+\xi$ and set $X(t) = x$ for all $T\leq t< T'$. Pick $j \in V$ according to $P_i = \lambda_i/\lambda$, and set $X(T')=\flip_j(x)$. Finally, we update $T \leftarrow T'$.
\end{enumerate}

The multinomial distribution $P_i = \lambda_i/\lambda$ in step 2 is the softmax, since $P_i \propto \exp(\sigma_iz_i/\tau)$. 
\subsection{Training}
\label{sec:training}

For brevity, let us define the following data statistics:
\begin{gather*}
\delta_i^{(n)} = x_i^{(n+1)}-x_i^{(n)}, \,\,\sigma_i^{(n)} = 1-2x_i^{(n)}, \\ 
z^{(n)}_i =\textstyle \sum_{ji\in E} w_{ji}x^{(n)}_j + b_i, \\
\lambda^{(n)}_i = \exp\left(\sigma_i^{(n)} z^{(n)}_i/\tau\right), \,\,\lambda^{(n)} =  \sum_i \lambda^{(n)}_i.
\end{gather*}
The log-conditional likelihood, $\mathcal{L}_\mathcal{D}(\theta)$, measuring similarity between model distribution $X(t)$ and $\mathcal{D}(t)$, amounts to:
\begin{align*}
&  \log p\! \left( X(t) = \mathcal{D}(t) \, , 0< t \leq t_{N+1} \, | \, X(0) = x^{(0)} \right) = \\
& \sum_{n=0}^{N-1} \log p\!\left(X(t)=\mathcal{D}(t) \, ,  t_n< t \leq t_{n+1}  \, | \, X(t_n)=x^{(n)}\right) \\
&= \sum_{n=0}^{N-1} \log \!\left\{\frac{\prod_i (\lambda_i^{(n)})^{\big\vert \delta_i^{(n)}\big\vert}}{\lambda^{(n)}}  \lambda^{(n)} \exp\big[{(t_{n}-t_{n+1})}\lambda^{(n)}\big]\right\} \\
& = \sum_{n=0}^{N-1}\left\{\mathcal{T}_n(\theta) + \mathcal{H}_n(\theta)\right\},
\end{align*}
with \emph{transition terms} $\mathcal{T}_n(\theta)$ and \emph{holding terms} $\mathcal{H}_n(\theta)$:
\begin{align*}
\mathcal{T}_n(\theta) &= \textstyle\sum_{i=1}^d \delta_i^{(n)} z_i^{(n)}/\tau, \\
\quad\mathcal{H}_n(\theta) &= -(t_{n+1}-t_n)\textstyle \sum_{i=1}^d \lambda^{(n)}_{i}.
\end{align*}
Since $-\mathcal{L}_\mathcal{D}(\theta)$ is convex (it consists of sums and exponentiations of linear/convex functions), we will adopt gradient methods for parameter optimization. Straightforward calculation yields the gradients of $\mathcal{T}_n(\theta)$ and $\mathcal{H}_n(\theta)$: 
\begin{align*}
\frac{\partial \mathcal{T}_n}{\partial w_{jk}} &= \phantom{-}\frac{1}{\tau} x_j^{(n)} \delta_k^{(n)}, 
\frac{\partial \mathcal{T}_n}{\partial b_{k}} = \phantom{-}\frac{1}{\tau} \delta_k^{(n)},\\
\frac{\partial \mathcal{H}_n}{\partial w_{jk}} &= -\frac{1}{\tau} x_j^{(n)} \sigma_k^{(n)} \lambda_k^{(n)} (t_{n+1}-t_n),  \\
\frac{\partial \mathcal{H}_n}{\partial b_k} &= -\frac{1}{\tau} \sigma_k^{(n)} \lambda_k^{(n)} (t_{n+1}-t_n).
\end{align*}
Note that these rules are highly local, in the sense that the update for a synaptic weight $w_{ik}$ depends only on the presynaptic state $x_i$, the postsynaptic rate $\lambda_k$, and the postsynaptic transition $\delta_k$. Similarly, the update for the bias $b_k$ depends only on the rate $\lambda_k$ and transition $\delta_k$. Moreover, any parameter update that agrees in sign with the above gradients will increase both the transition and holding terms of the log-conditional likelihood. Therefore, the training algorithm is robust to noise (up to sign) in the parameter updates. See SM C for the training algorithm pseudocode.

\section{Spike-Timing-Dependent Plasticity}
\label{sec:biology}

Recent work in neuroscience \cite{ermentrout2010mathematical} has shown that individual neurons often do not behave as single spiking units. In fact, different compartments of a neuron, such as the apical and basal regions of a pyramidal neuron, can spike apart from the soma.  To model such neural dynamics with a McCulloch-Pitts network, we represent the compartments of a neuron by the vertices of the graph $G$, and refer to them as \emph{units} rather than \emph{neurons} to avoid confusion. 

We assume that each spiking unit $i$ is either in an {\it armed} state $x_i=0$, capable of spiking, or in a {\it refractory} state $x_i=1$, representing a unit that just spiked and is incapable of spiking again until it recovers. We assume that a unit influences the spiking of a neighboring unit if and only if it is in the refractory state, a reasonable condition also assumed in \cite{pecevski2016learning}. In addition, to account for a well-known observation that the transmission of spikes across synapses is unreliable \cite{faisal2008noise}, we model the random chance of spiking events that depend on synaptic strengths $w_{ij}$ and unit biases $b_i$ as stochastic processes. Specifically, the wait time to the \emph{next spiking event} of unit $i$ is assumed to be exponentially distributed with rate
$$
\lambda_i(\theta) = r_i \exp\left\{\sigma_i z_i/\tau\right\},$$
where
$$\sigma_i =1-2x_i, z_i=\sum_{ji \in E} w_{ji}x_j + b_i,
$$
and where the parameter $r_i$ controls the background spiking rate of the unit when   $z_i$ vanishes. Note that when the unit is armed, the expected time to fire exponentially decreases with $z_i$; it is thus tempting to call $z_i$ the ``membrane potential" because of its similarity to classical integrate-and-fire neurons. However, spike events are \emph{not} stochastic when the membrane potential is known \cite{zador1998impact}. To avoid confusion, we shall therefore call $z_i$ the \emph{unsigned activity} and $\sigma_i z_i$ the \emph{signed activity}. Note also that while most neuron models assume a fixed refractory period, in our model it is stochastic, given by an exponential distribution whose expected time to recovery is inversely proportional to the exponential of $z_i$.

Learning in the network is implemented by two sets of update rules. Let $\eta$ be the learning rate. When a spike occurs, we apply transition updates:
\begin{align}
\label{eq:updates}
\begin{aligned}
\Delta w_{jk} &= \frac{\eta}{\tau} x_j^{(n)} \delta_k^{(n)}, & 
\Delta b_{k} &= \frac{\eta}{\tau} \delta_k^{(n)}.
\end{aligned}
\end{align}
On the other hand, when there are no spikes over a period of time, we apply holding updates:
\begin{align}
\label{eq:updates2}
\begin{aligned}
\Delta w_{jk}&= -\frac{\eta}{\tau} x_j^{(n)} \sigma_k^{(n)} \lambda_k^{(n)} (t_{n+1}-t_n), \\
\Delta b_k &= -\frac{\eta}{\tau} \sigma_k^{(n)} \lambda_k^{(n)} (t_{n+1}-t_n).
\end{aligned}
\end{align}
We may also interpret the holding updates as a fixed-rate decay of the weights and biases over time. 

In Bi and Poo's experiments \cite{bi1998synaptic} on spike-timing-dependent plasticity (STDP), the behavior of a pair of interconnected neurons, in which the synaptic weight is not strong enough for a presynaptic spike to cause a postsynaptic spike, is studied. The presynaptic neuron is initially stimulated to spike, and the amplitude of the excitatory postsynaptic current (EPSC) is measured. After five minutes, the presynaptic neuron is stimulated  again to spike. Next, after $\varepsilon$ time, the postsynapse is stimulated to spike. This pair of stimulations is repeated every second for 60 seconds. Finally, after twenty minutes, the amplitude of the EPSC is measured again and the percentage change is recorded. 

We prove that the average case behavior of the learning algorithm for an MPN agrees with the experimentally-observed synaptic potentiation and depression in biological neural networks. To simplify the mathematical analysis, we shall make the following assumptions:

1. Consider a synapse with presynaptic and postsynaptic neurons $i,j$ with states $x_i, x_j$. Let $w$ be the synaptic weight (directed weight from $i$ to $j$), and $b_i, b_j$ the neural biases. Let $z_i = b_i, z_j=w x_i + b_j$ be the unsigned activities. Let $\lambda_i = r_i \exp (\sigma_iz_i/\tau ), \lambda_j = r_j \exp(\sigma_jz_j/\tau)$ be the firing rates, where we set the background rate $r_i = r_j = r,$ and $\tau$ is the temperature controlling the degree of stochasticity. 

2. The neurons do not spike unless they are manually stimulated, but they may recover from their refractory state on their own accord.

3. The weight $w$ is updated according to equations~\eqref{eq:updates}  and ~\eqref{eq:updates2}, but the biases $b_i,b_j$ are fixed.

4. The ratio $\eta/\tau$ is much smaller than 1.

5. The refractory periods are on average shorter than the armed periods. 
\medskip
\begin{theorem}
\label{thm:stdp} Assuming the above conditions, let $\varepsilon$ be the timing of the presynaptic spike subtracted from that of the postsynaptic spike. If $\varepsilon$ is small and positive, then the expected synaptic weight change is
$$
\mathbb{E}[\Delta w] \approx C_1 e^{-\lambda_i|\varepsilon|}, \quad \lambda_i = r e^{-b_i/\tau},
$$
but if $\varepsilon$ is small and negative, then the expected synaptic weight change is
$$
\mathbb{E}[\Delta w]  \approx -C_2 e^{-\lambda_j|\varepsilon|}, \quad \lambda_j= r e^{-(w+b_j)/\tau },
$$
where $C_1, C_2$ are positive constants.
\end{theorem}
\begin{proof}
See SM B.
\end{proof}

\section{Experiments}
\label{sec:experiments}

In this section, we showcase the biologically plausible features and capabilities of the MPN\footnote{The codes are available at \url{https://github.com/owen94/MPNets}. }.
We first numerically verify the prediction of Theorem \ref{thm:stdp}, reproducing STDP-like curves from our learning rule. 
The MPN is then shown to be self-consistent, i.e. accurate inference of its own generative model is achievable. 
We next demonstrate the ability of MPN to robustly memorize spatiotemporal patterns. Namely, we show that our model can be trained to memorize repeating sequences of binary vectors with one-hop transitions and that these memories are robust in the sense that they are stable attractors of the MPN dynamics in the deterministic limit. 
When the dynamics are stochastic, we show that stochasticity assists in multiple-memory switching, as opposed to the deterministic counterpart whose dynamics eventually fixes at only one memory. 
Finally, we end this section by highlighting two potentially useful sequence learning applications: memorizing a long sequence of binary pictures (shown in Fig. \ref{fig: sequence} and the video in SM) and 
learning a generative model of experimental neural spike-train data.  The latter demonstrates that MPNs can effectively reproduce spike-timing statistics in neurobiology experiments.

\begin{figure}[t]
  \centering
   {\includegraphics[scale=.5]{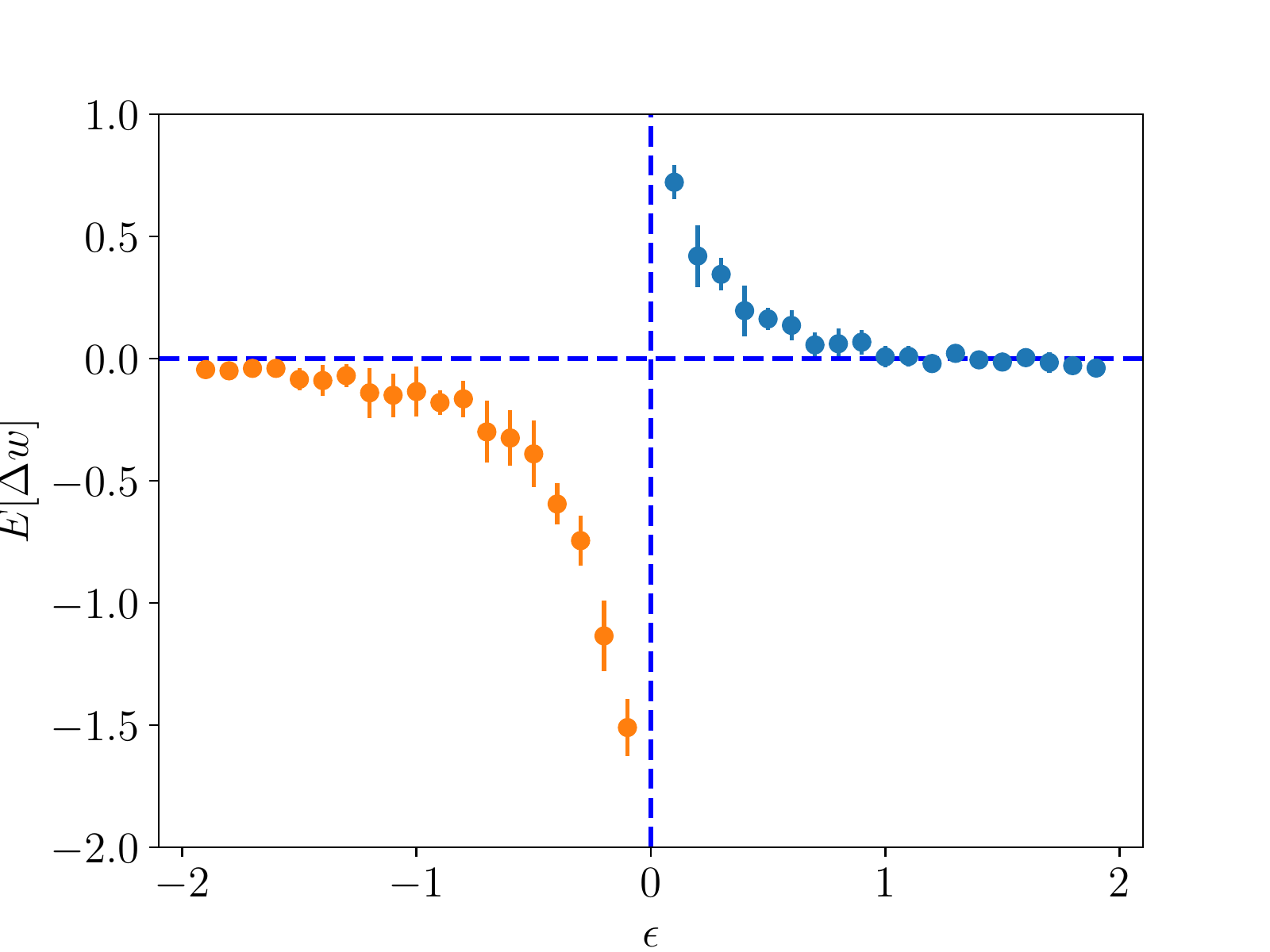}} 
   \caption{STDP-like curves emerge from simulating an MPN on a pair of neurons, in agreement with Theorem \ref{thm:stdp}. We denote the synaptic weight change as $\Delta w = \hat{w}_{ij} - w_{ij}$, and $\mathbb{E}(\Delta w)$ is computed from averaging 10 trials. Vertical dashed lines at each data point are error bars.}
  \label{stdp}
\end{figure}

\subsection{STDP and MPN Learning Rule }
\label{sec: STDP}
By simulating the dynamics of a presynaptic neuron $i$ and a postsynaptic neuron $j$, we now verify the emergence of STDP curves from our MPN learning rule. 
Following the algorithm discussed in the Spike-Timing-Dependent Plasticity section, 
the expected synaptic weight changes $\Delta w_{ij}$ can be measured as a function of the spike timing difference $\epsilon$. 

We simulate 10 trials for each $\epsilon$, while for each trial the updates of 60 consecutive spikes are accumulated, as suggested in \cite{bi1998synaptic}, to finally compute the new synaptic weight $\hat{w}_{ij}$. 
When there are no longer weight updates for every spike, the update will be terminated. 
Here, we set the learning rate for the transition and holding updates to $\eta_{\mathcal{T}} = 0.05$ and $\eta_{\mathcal{H}} = 0.001,$ respectively. Other parameters are $w_{ij} = 1, b_i = b_j = 0$ . 

Two initial conditions are considered: either neuron $i$ is in the refractory state and neuron $j$ is in the armed state or vice versa. 
Fig.\ref{stdp} displays simulation results for both initial conditions, which reveals that the expected (average) weight change consists of two exponential branches as a function of $\epsilon$, in agreement with Theorem \ref{thm:stdp}. This result is consistent with the STDP rule discussed in \cite{bi1998synaptic}, and the two branches do not need to be symmetric as assumed in \cite{scellier2017equilibrium}.

\begin{figure}
  \centering
  {\includegraphics[scale=0.15]{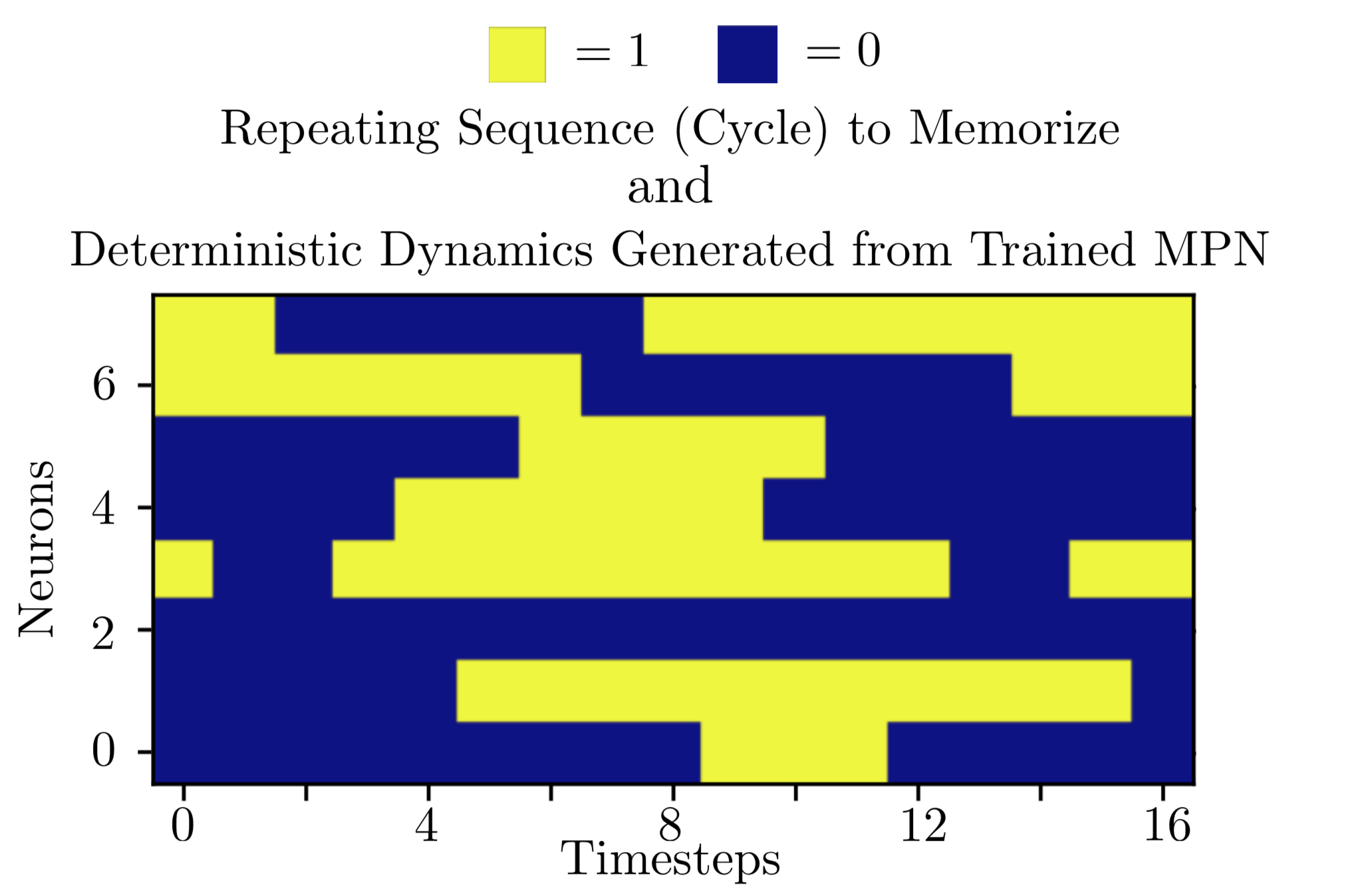}} 
  {\includegraphics[scale=0.25]{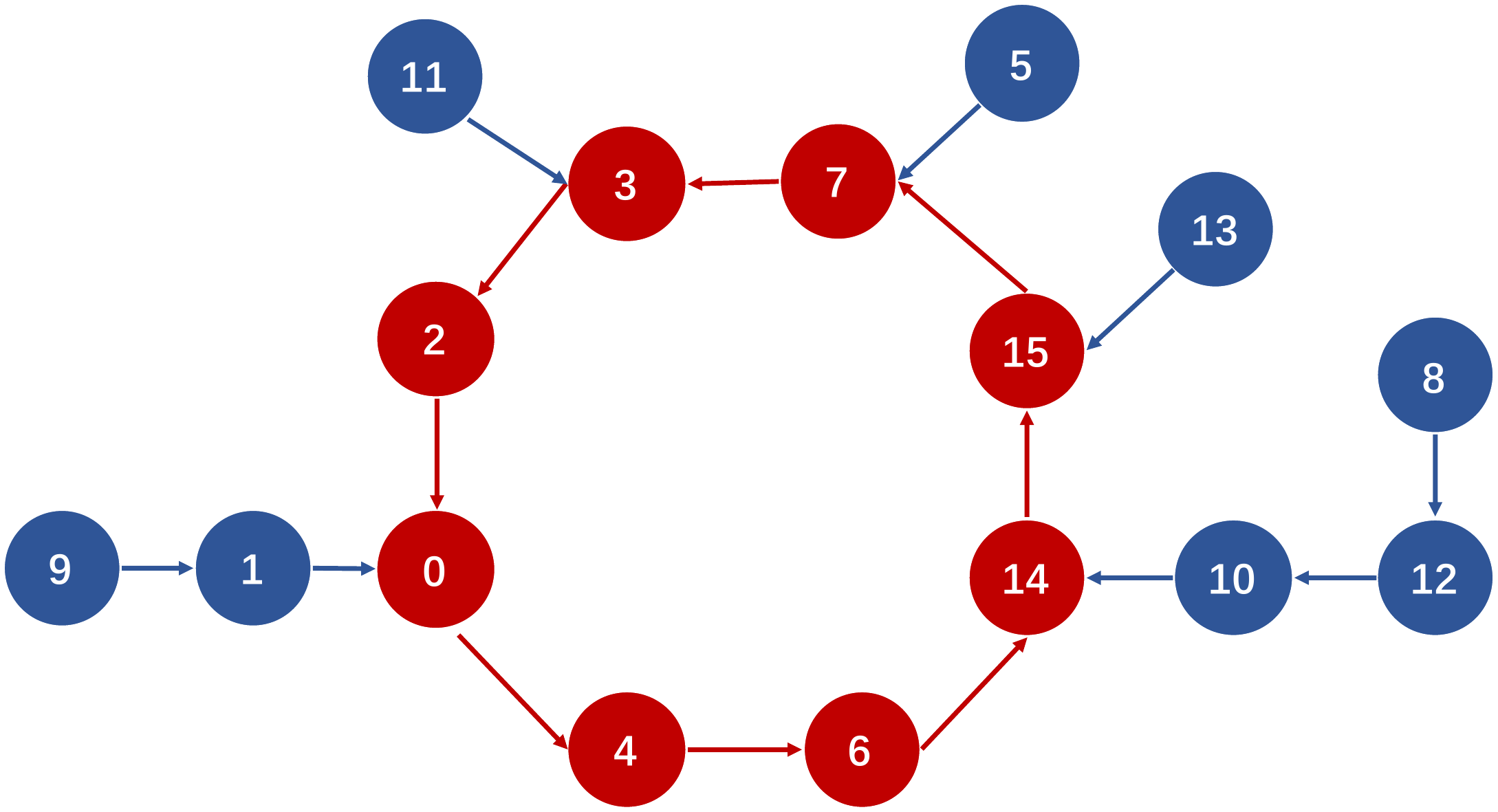}} 
   \caption{Deterministic dynamics of a trained MPN robustly memorize spatiotemporal binary patterns. For an MPN with 8 neurons, the cycle of period 16 to be memorized can be perfectly reproduced from the deterministic dynamics generated by the trained MPN {\bf (top)}. In addition, a spatiotemporal memory is {\it robust}, as illustrated in the smaller state space dynamics generated by another trained MPN with 4 neurons {\bf (bottom)}. The number on each node is the decimal representation of the corresponding state binary vector of length 4; e.g., $(1,0,1,0)$ is 5. Blue nodes are transient states that flow along the arrow toward the robust spatiotemporal memory, i.e. stable periodic attractor (red).}
  \label{cycles}
\end{figure}
\begin{figure*}[t]
  \begin{center}
  {\includegraphics[scale=.4]{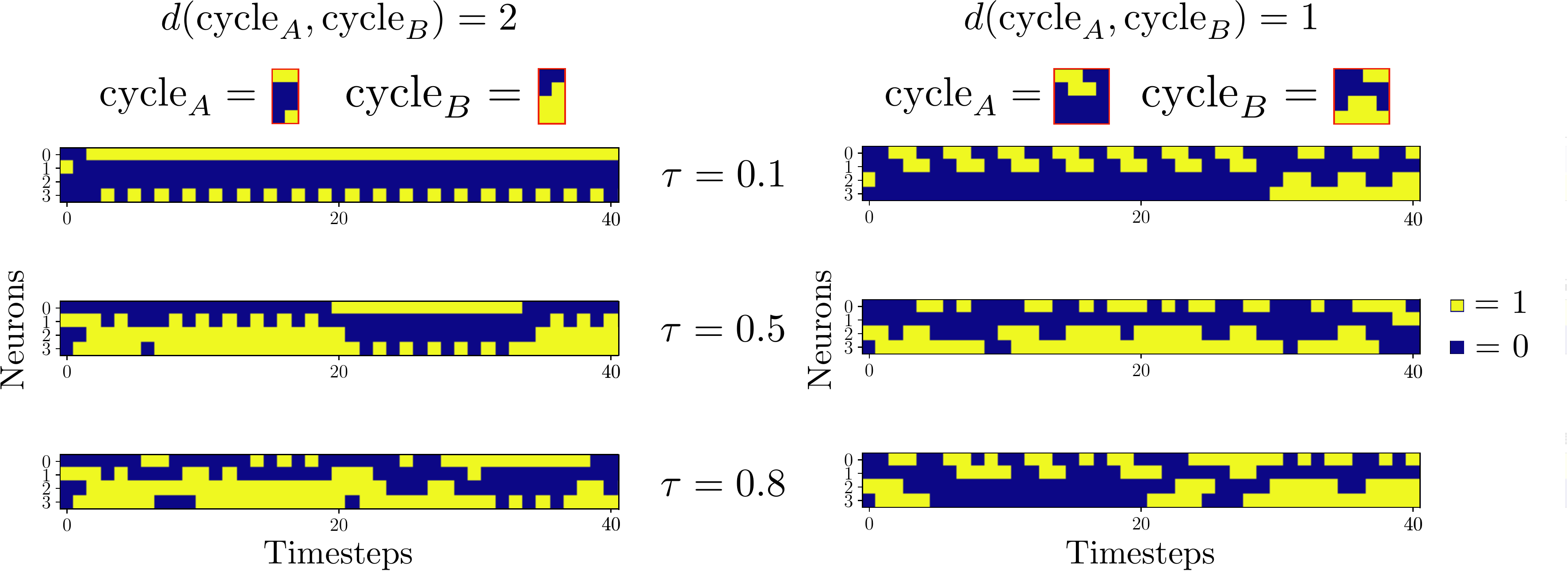}}
     \caption{40 states transitions under different noise parameter $\tau$ generated by an MPN with 4 neurons trained to memorize two cycles in the deterministic limit: \textbf{(left)} stochastic dynamics generated by an MPN that is trained to memorize cycles of period 2. \textbf{(right)} stochastic dynamics generated by an MPN that is trained to memorize cycles of period 4; see Fig. \ref{fig: 2cycles_topology} in SM for topologies. $d( \text{cycle}_A, \text{cycle}_B)$ denotes the minimal hamming distance between the two cycles, which measures how far away the stored spatiotemporal memories are. The larger the $d$, the more rare bit flips (induced by intrinsic noise) are required to transit to the other memory, and hence memory switching becomes less likely when $d$ is large. At the same $\tau$, we thus observe less frequent memory switching in \textbf{(left)} than in \textbf{(right)}.}
  \label{fig: 2cycles}
  \end{center}
\end{figure*}

\textbf{Consistency Check:} Here, we numerically show
that MPNs can consistently learn the parameters of another MPN that generates training samples. 
100,000 samples were generated from the pre-defined MPN of 10 neurons with weights $w$ and biases $b$. Next, we train an MPN to fit the generated data. 
For the reconstruction error using $\ell_1$ norm defined by $error(w, \hat{w}) = \frac{1}{100}\| w - \hat{w} \|$, we find $error(w, \hat{w}) = 0.03$ after 30 epochs of training. This consistency check shows that MPN is self-consistent. Note also that the training algorithm converges fast, plateauing  after 5 training epochs (see SM D.1).

\subsection{Robust Spatiotemporal Memory}
\label{sec: robustmemory}
\textbf{Spatiotemporal memory:} The advantage of weight asymmetry is the flexibility to memorize repeating {\it sequences} of spikes ({\it cycles}) regarded as spatiotemporal memories. 
Here, we show that MPN can be trained to memorize a cycle using the deterministic update rule (the noiseless limit is identical to setting $\tau = 0$). This update rule amounts to flipping the neuron $i$ with the highest rate, i.e., $i = \arg \max \lambda_{j}$. Fig.\ref{cycles} (top) shows the cycle of period 16 to memorize, which is the repeating sequence to be stored by an MPN with 8 neurons. After training with this sequence, MPN can successfully generate the same sequence (reproducing Fig.\ref{cycles} (top)) from the dynamics in the limit $\tau = 0,$ provided the initial states of the spikes are identical.  Hence, MPN can store the sequence as a spatiotemporal memory.

\textbf{Robustness:} This repeating sequence memorized by the trained MPN with 8 neurons is in fact the stable attractor of the entire state space dynamics generated by the trained MPN. Namely, any $2^8 = 256$ initial binary states will eventually flow toward the cycle in Fig.\ref{cycles} (top). Thus, this spatiotemporal memory is {\it robust}.
Fig.\ref{cycles} (bottom) illustrates robust spatiotemporal memory in the state space dynamics generated by the trained MPN with 4 neurons, where the entire state space consists of $2^4 = 16$ elements. All transient states (blue) are in the basin of attraction of the spatiotemporal memory (red), which is the repeating sequence $[0, 4, 6, 14, 15, 7, 3, 2, 0]$.

  \begin{figure*}[t]
    \begin{center}
  {\includegraphics[scale=0.5]{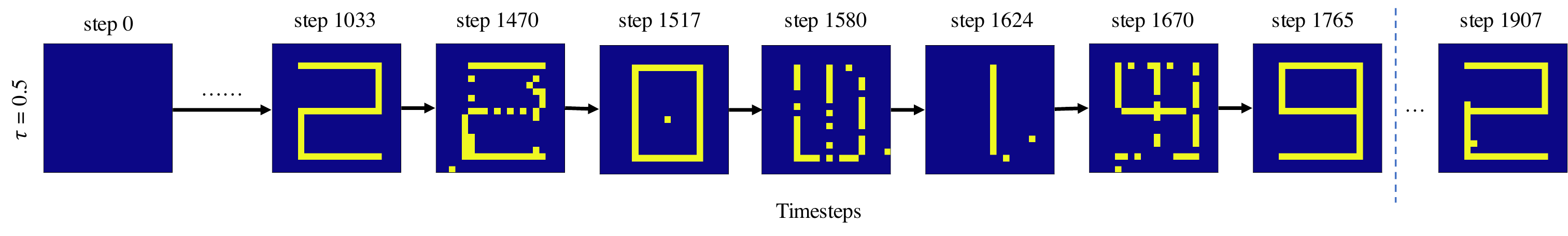}}
  {\includegraphics[scale=0.5]{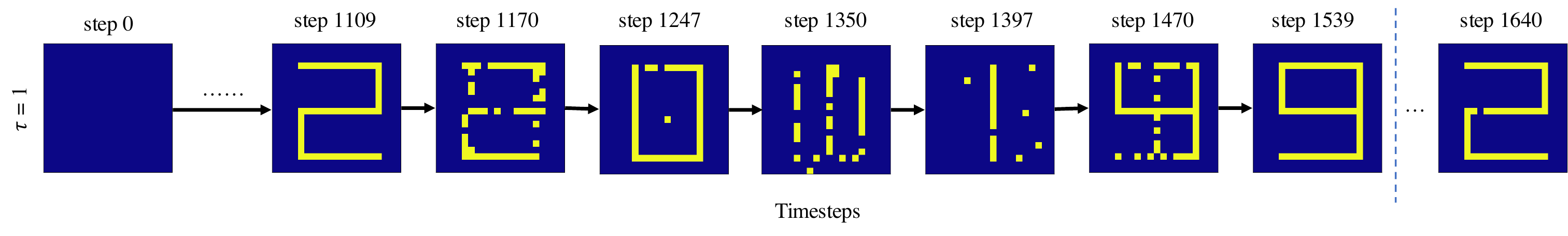}}
     \caption{Spatiotemporal pattern completion: dynamics generated by the trained MPN with $\tau=0.5$ and $\tau=1$ can robustly recover the multi-hop away memory sequence 2019. Results for $\tau = 2$ can be found in the SM.}
  \label{fig: sequence}
  \end{center}
\end{figure*}

\begin{figure*}[t]
\begin{center}
   {\includegraphics[scale=0.2]{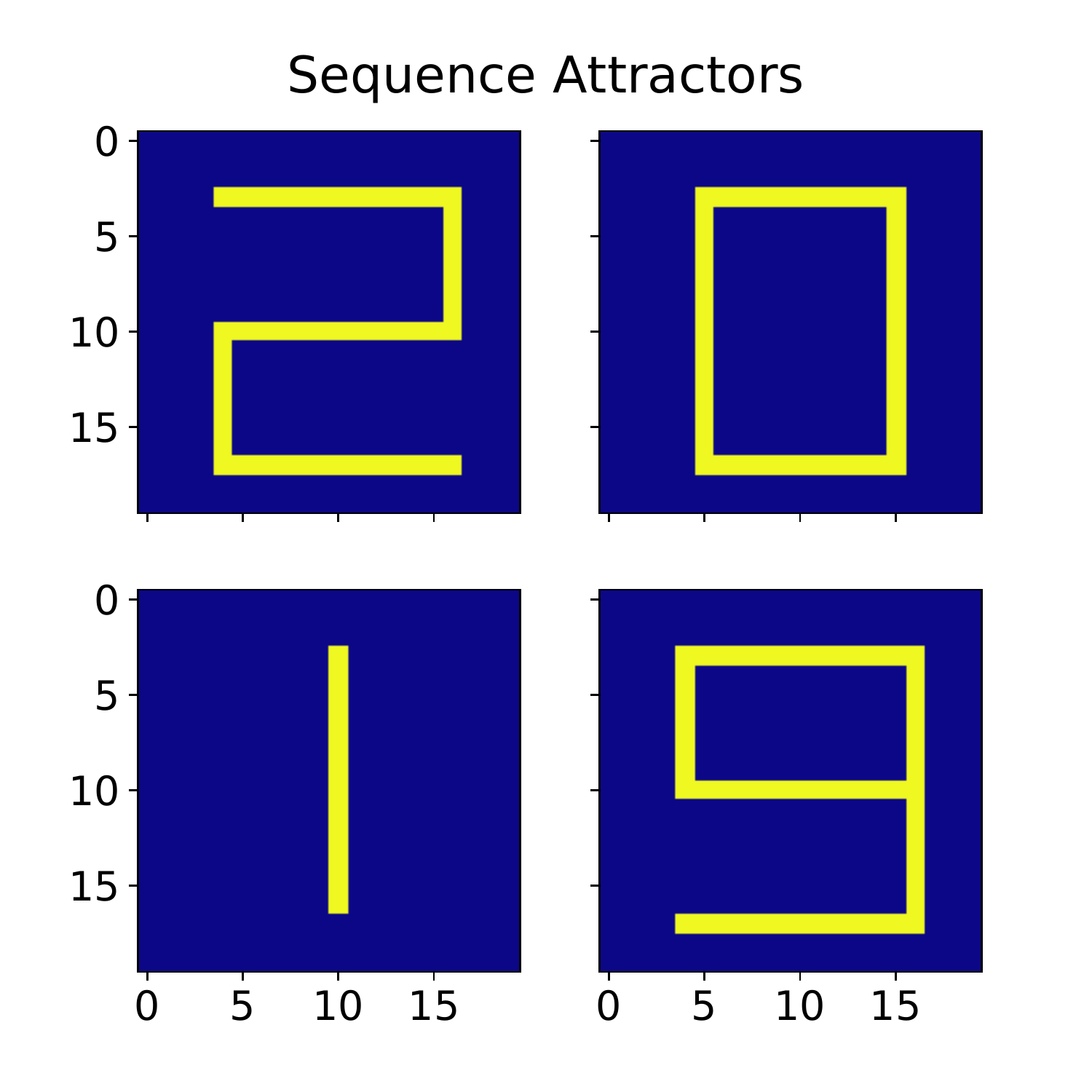}}
    {\includegraphics[scale=0.21]{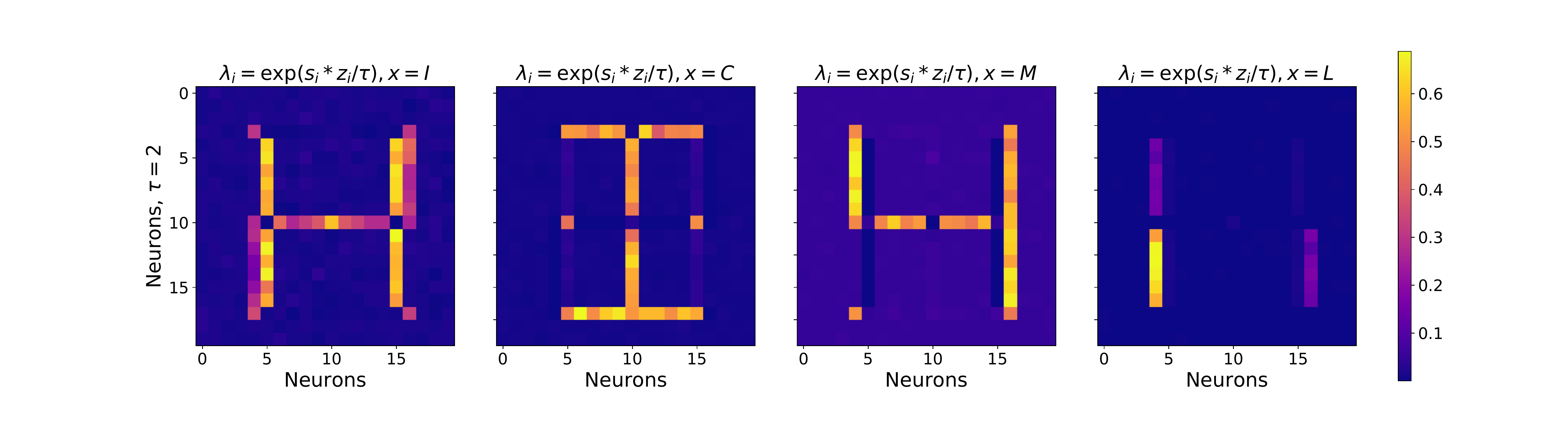}}
  \caption{\textbf{(Left)}. Multi-hop memory sequence to store that proceeds from ``2" to ``0" to ``1" and to ``9". \textbf{(Right)}. Rates $\lambda_i$ computed for different input spiked patterns; left to right: ``2", ``0",``1" and ``9".  Larger $\lambda_i$ indicates a higher probability for the neuron $i$ to flip. }
  \label{fig: multi-hop sequence}
  \end{center}
\end{figure*}

\textbf{Multiple Memories and the Role of Noise:} It is possible to train an MPN to store multiple spatiotemporal repeating patterns, analogous to standard Hopfield nets that can store multiple static patterns.
Here, we demonstrate the simplest scenario where 2 non-overlapping cycles are stored as robust spatiotemporal memories in a small network with 4 neurons. Generalization to multi-cycle storage in larger networks is straightforward and will be reported elsewhere. 

In these trainings, the first cycle to be memorized, say $\text{cycle}_A$, is selected to train an MPN.
After the training is completed, we select another one, $\text{cycle}_B$, that is disjoint from $\text{cycle}_A,$ and train the MPN that was originally trained on $\text{cycle}_A$. Eventually, the MPN can store {\it both} $\text{cycle}_A$ and $\text{cycle}_B$. 
We illustrate two trained MPNs: one that stores 2 cycles of period 2, and another one that stores 2 cycles of period 4. 
The deterministic state transition diagram generated by these MPNs is shown in Fig.\ref{fig: 2cycles}; the state space consists of two cycles of period 2 in Fig.\ref{fig: 2cycles}(a) 
which are $\text{cycle}_A = \{(1,0,0,0), (1, 0, 0, 1) \}$ and $\text{cycle}_B = \{(0,0,1,1), (0, 1, 1, 1) \}$, 
and two cycles of period 4 in Fig.\ref{fig: 2cycles}(b), which are  
$\text{cycle}_A = \{(1,0,0,0), (1, 1, 0, 0), (0, 1, 0, 0), (0, 0, 0, 0) \}$ and $\text{cycle}_B = \{(0,0,0,1), (0, 0, 1, 1), (0,1,0,0), (0,1,1,0) \}$. This notation represents an ordered set, in which a sequence generated by the trained MPN proceeds from left to right.

In the presence of noise ($\tau > 0$), fluctuations may destroy the robustness of cycles. Namely, rather than remaining in one of the  cycles (deterministically stable periodic attractors), states could be driven away from one robust cycle to another through a series of rare one flip transitions.  Fig. \ref{fig: 2cycles} shows the response of the trained MPNs as the strength of intrinsic noise $\tau$ increases. 
Thus, one useful role of noise in MPN is to facilitate transitions between stored spatiotemporal patterns of spikes that would not otherwise be possible in the deterministic limit. Whether one could exploit this intrinsic noise-induced memory switching phenomena to design controllable memory switching networks is worthy of future investigation.

\subsection{Application: Robustly Memorizing Sequences of Binary Pictures}

Robust spatiotemporal memory is an appealing feature of the MPN; however, successive elements of a preferable sequence to store might not be one-hop different as assumed in the previous section. For instance, the successive elements in the sequence of digits ``2019" represented as binary pictures of Fig. \ref{fig: multi-hop sequence} (left) differ by multi-hop transitions.
Here, we show that MPNs can be trained to learn sequences of pictures that differ by multi-hop transitions, illustrated in Figs.\ref{fig: sequence}- \ref{fig: multi-hop sequence}, and also in the video in SM.

The multi-hop away sequence to learn consists of [the all-zero state, ``2", ``0", ``1", ``9"], in this specific order. Each static binary picture in the sequence is represented as a 20 $\times$ 20 binary matrix.
 To transit among these static memories in the correct order, we randomly assign 100 one-hop transitions between every successive pair of the static memories; i.e., the length for the entire training sequence, including the transitions from the all-zero state to ``2", ``2" to ``0", ``0" to ``1" and ``1" to ``9", is 400. The learning rate is set to $\eta = 0.001$. 
For details on training procedure, see SM D.3.

Fig.\ref{fig: sequence} shows snapshots of the dynamics at two different noise parameters $\tau$ generated from the trained MPN.
Initialized with the all-zero state, the trained MPN almost perfectly reproduces the multi-hop sequence ``2019" with multiple one-hop transitions that connect consecutive elements.

This generated sequence is also robust to noise, in the sense that we can recognize the sequence of static memories ``2", ``0", ``1", and ``9" as time evolves for different noise levels. 
As a consequence of our training procedure, the trained MPN will generate the cycles containing multi-hop away memories in the correct order; i.e., the dynamics will loop ``2" $\rightarrow \cdots \rightarrow$ ``0"  $\rightarrow \cdots \rightarrow$ ``1"  $\rightarrow \cdots \rightarrow$ ``9"  $\rightarrow \cdots \rightarrow$ ``2" $\rightarrow \cdots$ (see Fig.\ref{fig: sequence}). 

We further examine the rate $\lambda_i(\theta) = r_i \exp\left\{\sigma_i z_i/\tau\right\} $ to illustrate that the MPN indeed learned the sequence. Given the weights and biases of the trained MPN, we compute the rates when the static memories are the input. Since the neuron with a higher rate $\lambda_i$ will be more likely to flip next, we can estimate the preferable transitions that will take place given $\lambda_i$. 
For instance, when feeding the trained MPN with the input ``2", most rates associated with the spiked neurons of the memory ``0" and ``2" become significanly higher than the rates at the other neurons, see Fig.\ref{fig: multi-hop sequence}. This is consistent with the fact that ``2" will transit to ``0" through multi-hop transitions and that the spiked neurons of ``2" will recover with higher rates than those neurons remaining in the armed state.  However, the neurons that spike at {\it both} the memory ``2" and ``0", i.e. the bottom horizontal line, will prefer to not flip, consistent with the almost negligible rates at the neurons that both memories already share.   

Note finally that while Hopfield networks can be used for recalling or denoising static memories \cite{hopfield1982neural,hillar2018robust}, we have shown that MPNs can be applied to recall spatiotemporal memories or complete spatiotemporal patterns.
Thus, the MPN with asymmetric weights can be regarded as an STDP-consistent and non-equilibrium generalization of the dynamics arising from the symmetric Hopfield network. 
Regarding weight asymmetry, we observe from the trained MPN that the learned weights are considerably asymmetric. The asymmetry measure $|w_{ij} - w_{ji}| /( (|w_{ij}| + |w_{ji}|)/2) = 1.38, 1,53, 1.58$ for training parameters  $\tau = 0.5, 1, 2$, respectively. 


\subsection{Application: Inferring Generative Models for Neural Spike-train Data}

We now apply the MPN to learn spike-train statistics of a neurobiology experiment, which reports the neuronal activities of cat primary visual cortex \footnote{The dataset is available on \url{https://crcns.org/data-sets/vc/pvc-3}}. The experiment recorded spike-train dataset of 25 neurons. We preprocess the data and form 60288 one-hop transitions, of which $70\%$ (42202) is used for training and $30\%$ (18086) for testing. 
The learning rate is set to $\eta = 0.01$.


To see how well the MPN reproduces the spike-timing statistics of the data, we generate spike trains (from the MPN trained at $\tau = 1$) that last for the same duration as that of the test data, starting from the same initial spike state.
Fig.\ref{lld} shows that the distribution of the inter-spike intervals (ISI) of the generated spikes is in good agreement with that of the test data. 
The visiting frequencies of generated spike states are also in agreement with that of the test data as well (see Fig.6 in SM D.4). 

We now compare the performance of MPN in reproducing the ISI statistics to that of the generalized linear model (GLM), a standard statistical model of ISI, using the Kullback-Leibler(KL) divergence as a performance metric. We first bin the ISI into 20 slots, compute the density of each bin, and obtain the distributions $P_{test}(ISI), P_{MPN}(ISI), P_{GLM}(ISI)$. Then, the KL-divergences are calculated from $KL(MPN) = KL(P_{MPN}(ISI) \| P_{test}(ISI))$ and $KL(GLM) = KL(P_{GLM}(ISI) \| P_{test}(ISI))$, where $KL(p\|q) = - \sum_{x} p(x) \log (p(x)/q(x))$. Remarkably, we found that $0.21 \pm 0.11 = KL(GLM) >  KL(MPN) = 0.04 \pm 0.003,$ where the standard deviations are calculated from 10 trials. Thus, MPNs can capture the statistics of ISI even more accurately than the ubiquitous GLM.
\begin{figure}[t]
  \centering
   {\includegraphics[scale=0.5]{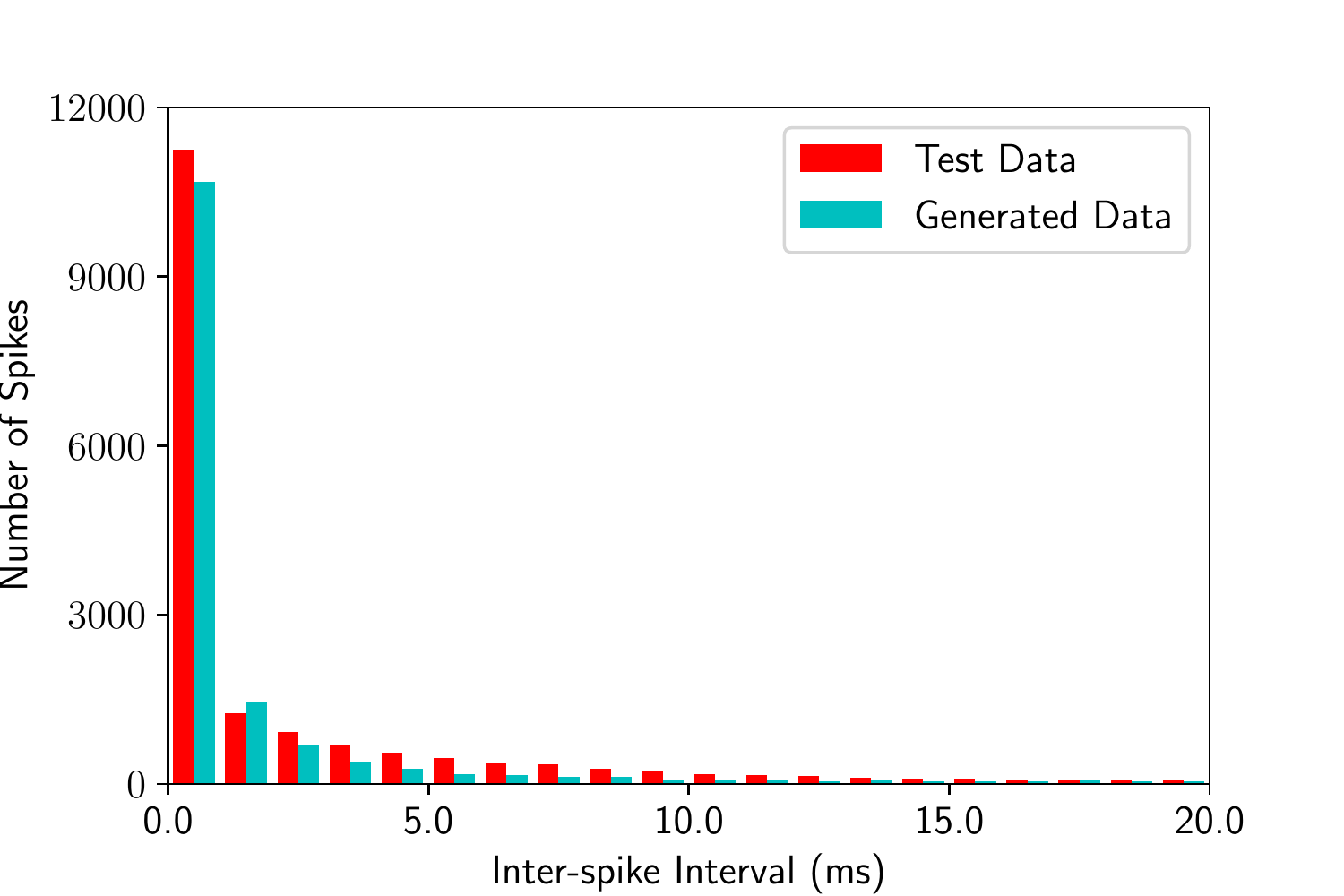}}
   \caption{Histogram of the ISI of the spikes generated from the trained MPN is in excellent agreement with that of the test data.}
  \label{lld}
\end{figure}

\section{Related Work}
\label{sec:related_work}

While sequence learning tasks have been studied in the literature \cite{amari1972learning,brea2013matching,memmesheimer2014learning}, the work most related to MPN is \cite{brea2013matching}, which focuses on training spiking neural networks with biologically plausible rules. While both \cite{brea2013matching} and this work attempt to learn the distribution of sequences, they differ in the model, the learning rule, and the applications. 
Regarding model assumptions, \cite{brea2013matching} assumes a firing probability to be sigmoidal during $T$ discrete-time bins; MPN, however, assumes a continuous-time neuron model with exponentially distributed refractory periods 
and stochastic neurons. Such different models and assumptions lead to distinct learning rules.
Finally, while \cite{brea2013matching} exploits hidden units to learn a sequence, the MPN reported here consist of only visible neurons. 
We expect that, by adding hidden neurons, MPN will become even more expressive. We leave hidden neurons problem to be explored in future work.  

Biological plausibility \cite{rezende2011variational,bengio2015towards,bengio2017stdp,scellier2017equilibrium,liu2018variational} has also been emphasized in the machine learning community, with the hope to uncover efficient learning algorithms inspired by neuroscience. However, many proposed STDP-consistent algorithms assume symmetric weights, a constraint generally considered biologically implausible. For example,  \cite{scellier2017equilibrium} introduces a novel learning paradigm termed Equilibrium Propagation that offers a biologically plausible mechanism for backpropagation and credit assignment in Deep Learning. The learning algorithm implements a form of STDP for neural networks with symmetric weights, i.e. $w_{ij} = w_{ji}$, such that, in one of the learning phases, the weight update rule satisfies
$\frac{\partial w_{ij}}{\partial t} \propto \rho(u_i) \frac{d\rho(u_j)}{dt}+\rho(u_j) \frac{d\rho(u_i)}{dt}$, where $u_i$ and $\rho(u_i)$ are the membrane potential and its firing rate. In contrast, our model possesses a different weight update rule (see Spike-Timing-Dependent Plasticity section), 
does not pose the symmetric weight condition, and can reproduce asymmetric STDP-like curves (see the STDP and MPN Learning Rule section).  

Lastly, our learning algorithm is partly inspired by Minimum probability flow (MPF) \cite{Sohl-DicksteinICML,hillar2012efficient}, a symmetric-weight discrete-time algorithm that efficiently learns distributions of static variables. 
Although the update rule of MPF is akin to the the transition updates of an MPN with symmetric weights, an asymmetric-weight continuous-time MPN can also learn distributions of sequences. 

\section{Conclusion}
\label{sec:conclusion}

We introduce a biologically plausible continuous-time sequence learning algorithm for spiking neural networks and demonstrate its capability to learn neural spike-train data as well as to robustly store and recall spatiotemporal memories. The hallmark of our local learning rule is that it provides a normative explanation of STDP. There are several directions for future investigations. Firstly, akin to the Hopfield network capacity problem,  we could explore the capacity of MPNs for storing random repeating sequences. More interestingly, given a set of transitions between states that is not necessarily a cycle, is it feasible to train the network to memorize those transitions? The answer could hint at whether MPNs possess large long-sequence storage capacity, a highly sought-after characteristic of modern sequence learning models. Furthermore, while the current MPN is a fully-visible network, extensions to a deep network with hidden layers for large-scale computer vision and natural language tasks would be worth investigating. Lastly, implementing the local learning rule on neuromorphic computing hardware could be a near-term application of our neuro-inspired framework. 

\section{Acknowledgement}
\label{sec: acknowledgement}
Zuozhu Liu is supported by  the Singapore MOE grant MOE2016-T2-2-135 under PI Alexandre Thiery. Thiparat Chotibut and Shaowei Lin are supported by SUTD-ZJU grant ZJURP1600103. We thank Andrew Tan, and Andreas Herz for helpful discussions. We also thank the reviewers of the 34th AAAI Conference on Artificial Intelligence for their helpful comments on this work.
\bibliography{example_paper}
\bibliographystyle{aaai}

\begin{appendix}

\section{Supplementary Material (SM)}

\section{A: Rate Scaling}
\label{app:scaling}

To construct a physically consistent model, the rate $\lambda_i$ introduced in the McCulloch-Pitts Networks section of the main text must have the unit of the inverse of time. This can be done by introducing the parameter $r_i$ with an appropriate unit: 
\begin{align*}
	\lambda_i &= r_i\exp(\sigma_i z_i /\tau)\\
	&=\exp[\sigma_i z_i /\tau + \ln(r_i)].
\end{align*}
We now show that introducing such an extra parameter is amount to correcting $b_i$ and $w_{ii}$ by appropriate $r_i$-dependent factors. 
Consider the exponent $\sigma_i z_i /\tau + \ln(r_i).$ Since $\sigma_i^2 = (1-2x_i)^2 = 1$ for a binary neuron's state $x_i \in \{0,1\},$ 
the exponent can be rewritten as:
\begin{align*}
\sigma_i z_i /\tau + \ln(r_i)  &= \sigma_i \Big(z_i + \sigma_i\tau\ln (r_i) \Big)/\tau \\
&= \sigma_i \tilde z_i /\tau,\\
\end{align*}
where the new input activity $ \tilde z_i $ is related to the input activity $z_i$ by the shift:
\begin{equation*}
	\tilde z_i = z_i + \sigma_i \tau \ln(r_i).
\end{equation*}
In terms of the weights and biases, the new activity reads:
\begin{align*}
\tilde z_i &= \left(\sum_{j} w_{ji}x_j + b_i\right) + (1-2x_i)\tau \ln(r_i)\\
&= \sum_{j} \tilde w_{ji}x_j + \tilde b_i,
\end{align*}
where the new weights and biases are given by:
\begin{align*}
 \tilde w_{ij} &= \begin{cases}
    w_{ij}, & \text{if $i \neq j$}.\\
    w_{ii} - 2\tau\ln (r_i) , & \text{otherwise}.
  \end{cases}\\
\tilde b_i &= b_i+\tau \ln (r_i). \\
\end{align*}

Thus, with appropriate redefinition of $w_{ii}$ and $b_i$ in mind, we can simply work with $\lambda_i = \exp(\sigma_i \tilde z_i /\tau)$ as introduced in the main text without loss of generality. 
\section{B: Proof of Spike-Timing-Dependent Plasticity}
\label{app:proofstdp}

We now prove Theorem \ref{thm:stdp} 
\begin{proof}
First, we compute the expected synaptic update to the directed weight $w$ that sends the signal from the presynaptic neuron $i$ to the post synaptic neuron $j$, provided the initial states at time $t=0$ are $x_i = x_j = 1$, and the initial synaptic weight is $w_1$. In this scenario, we will show that the expected synaptic update is $-(\eta/\tau)\varphi(w_1)$ for some function $\varphi$ to be determined; this result will later later used. We assume that neurons do not spike again after they recover. Note that after the presynaptic neuron recovers, there are no other synaptic updates because $x_i =0$. 

Let us first define the firing rates which will appear frequently in the calculation:
\begin{gather*}
\gamma_{i} = re^{b_{i}/\tau}, \quad \tilde{\gamma}_{i} = r/\gamma_{i},\\
\gamma_{j1} = re^{(b_{j}+w_1)/\tau}, \quad \tilde{\gamma}_{j1} = r/\gamma_{j1}.
\end{gather*}
These rates are defined given the initial conditions $x_i = x_j = 1.$ Then, there are two cases to consider, as outlined below. For each case, we derive the resulting synaptic update. 
\begin{enumerate}
\item[C1.] The presynaptic neuron recovers {\it before} the postsynaptic neuron. The probability of this event occurring is $\tilde{\gamma}_{i} / (\tilde{\gamma}_{i}+\tilde{\gamma}_{j1})$. Let $\varepsilon_i$ be the time of the presynaptic recovery, then this holding time will be sampled from 
\begin{align*}
\varepsilon_i \sim \text{Exp}(\tilde{\gamma}_{i}+\tilde{\gamma}_{j1}). 
\end{align*}
From Eq. (2) in Sec. 3 in the main text, the holding update reads
\begin{align*}
\Delta w &= (\eta/\tau) \varepsilon_i \tilde{\gamma}_{j1}.
\end{align*}
\item[C2.] The presynaptic neuron recovers {\it after} the postsynaptic neuron recovers. The probability of this event occurring is $\tilde{\gamma}_{j1} / (\tilde{\gamma}_{i}+\tilde{\gamma}_{j1})$. Let $\varepsilon_j$ be the time of the presynaptic recovery and $\varepsilon_j+\varepsilon_i$ be that of the postsynaptic recovery.  In this case, we have the following:
\begin{align*}
\varepsilon_j &\sim \text{Exp}(\tilde{\gamma}_{i}+\tilde{\gamma}_{j1}),\\
w_2 &= w_1 + (\eta/\tau) \varepsilon_j \tilde{\gamma}_{j1}-  (\eta/\tau),\\
\gamma_{j2} &= re^{(b_{j}+w_2)/\tau},\quad\tilde{\gamma}_{j2} = r/\gamma_{j2},\\
\varepsilon_i &\sim \text{Exp}(\tilde{\gamma}_{i}),
\end{align*}
and the weight update reads
\begin{equation*}
\Delta w = (\eta/\tau)[\varepsilon_j \tilde{\gamma}_{j1}  -\varepsilon_i  \gamma_{j2}-1].	
\end{equation*}

\end{enumerate}
Consequently, combining the two cases, the expected synaptic update is
\begin{align*}
&(\eta/\tau)\tilde{\gamma}_{i} \tilde{\gamma}_{j1}\int_0^\infty \!\!\varepsilon_i  e^{-(\tilde{\gamma}_{i}+\tilde{\gamma}_{j1})\varepsilon_i} \,d\varepsilon_i 
+(\eta/\tau)\tilde{\gamma}_{i}\tilde{\gamma}_{j1} \cdot \\ &\int_0^{\infty}\!\!e^{-(\tilde{\gamma}_{i}+\tilde{\gamma}_{j1})\varepsilon_j} \!\int_0^{\infty}\!\!(\varepsilon_j\tilde{\gamma}_{j1}  -\varepsilon_i  \gamma_{j2}-1) e^{-\tilde{\gamma}_{i}\varepsilon_i} d\varepsilon_i d\varepsilon_j,
\end{align*}
which simplifies to $-(\eta/\tau)\varphi(w_1)$, where 
$$
\varphi(w_1) \equiv \frac{\gamma_{i}e^{-\eta/\tau}}{\tilde{\gamma}_{i}+(1-\eta/\tau)\tilde{\gamma}_{j1}} > 0.
$$
Note that in the limit $\eta/\tau \ll 1$, we get the approximation
$$
\varphi(w_1) \approx \frac{\gamma_{i}}{\tilde{\gamma}_{i}+\tilde{\gamma}_{j1}}.
$$
Suppose the refractory periods are shorter on average than the armed periods as assumed in the main text, we deduce that $\gamma_{i}$ is smaller than $\tilde{\gamma}_{i}$ and $\tilde{\gamma}_{j1}.$ Hence, $\varphi(w_1) < 1$.

Using these results, we will now derive the synaptic update when the presynaptic spike occurs at $t=0$ and the postsynaptic spike {\it follows} at $t=\varepsilon$, the situation neccesary for computing the expected synaptic weight in one branch of the STDP-like curves. 

Define the following firing rates 
\begin{equation*}
\gamma_{j} = re^{(b_{j}+w)/\tau}, \tilde{\gamma}_{j} = r/\gamma_{j},	
\end{equation*}

Assume the presynaptic recovery time is also exponentially distributed, with the rate $\tilde{\gamma}_{i}$. There are two cases to consider here.
\begin{enumerate}
\item[C1.] The presynaptic neuron recovers at time $t=\varepsilon_i$ before the postsynaptic spike. This event occurs with probability $1-e^{-\tilde{\gamma}_i\varepsilon}$. The probability density for $\varepsilon_i$ to be in $(0,\varepsilon)$, and the corresponding synaptic update are:
\begin{align*}
p(\varepsilon_i) &= \tilde{\gamma}_i e^{-\tilde{\gamma}_i \varepsilon_i} /(1-e^{-\tilde{\gamma}_i\varepsilon}),\\
\Delta w & =  -(\eta/\tau) \varepsilon_i \gamma_j.
\end{align*}
\item[C2.] The presynaptic neuron recovers after the postsynaptic spike. The probability of this event occurring is $e^{-\tilde{\gamma}_i\varepsilon}$. Thus, at $t=\varepsilon$, we have $x_i=x_j=1$, and
\begin{align*}
w_1 &= w- (\eta/\tau) \varepsilon \gamma_j  + (\eta/\tau).
\end{align*}
Consequently, the expected synaptic update is $(\eta/\tau)[-\varphi(w_1)-  \varepsilon \gamma_j  +1]$.
\end{enumerate}
Combining both cases, the expected synaptic update when the presynaptic neuron spikes {\it before} the postsynaptic neuron simplifies to:
$$
\mathbb{E}[\Delta w] =(\eta/\tau)e^{-\tilde{\gamma}_i\varepsilon}[1+ \gamma_i \gamma_j(1-e^{\tilde{\gamma}_i\varepsilon})-\varphi(w_1)].
$$
In the limit $\varepsilon \ll 1$, we get the desired approximation:
$$
\mathbb{E}[\Delta w] \approx C_1 e^{-\tilde{\gamma}_i\varepsilon},\quad C_1= (\eta/\tau)[1-\varphi(w_1)] > 0.
$$
This gives the positive branch of the STDP-like curves.

Finally, to study the other branch of the STDP-like curves, suppose the presynaptic spike occurs at $t =\varepsilon$ {\it after} the postsynaptic spike at $t=0$. Assume, as usual, the postsynaptic recovery time is exponentially distributed, with rate $\tilde{\gamma}_j$. There are two cases to consider.
\begin{enumerate}
\item[C1.] The postsynaptic neuron recovers at time $t=\varepsilon_j$ before the presynaptic spike. This event occurs with probability $1-e^{-\tilde{\gamma}_j\varepsilon}$. The synaptic update is zero for $0 < t < \varepsilon$ because  $x_i=0$, and for $\varepsilon < t < \varepsilon+\varepsilon_i$ where $\varepsilon_i$ is the presynaptic recovery period, the update is nonzero.  Thus,
\begin{align*}
\varepsilon_i &\sim \text{Exp}(\tilde{\gamma}_{i}),\\
\Delta w & =  -(\eta/\tau) \varepsilon_i  \gamma_{j}.
\end{align*}
\item[C2.] The postsynaptic neuron recovers after the presynaptic spike. The probability of this event occurring is $e^{-\tilde{\gamma}_j\varepsilon}$. Thus, at $t=\varepsilon$, we have $x_i=x_j=1$ and
\begin{align*}
w_1 &= w+ (\eta/\tau) \varepsilon \tilde{\gamma}_j .
\end{align*}
Consequently, the expected synaptic update is $(\eta/\tau)[-\varphi(w_1)+\varepsilon \tilde{\gamma}_j]$.
\end{enumerate}
Combining both cases, the expected synaptic update when the presynaptic neuron spikes {\it after} the postsynaptic neuron simplifies to:
\begin{align*}
\mathbb{E}[\Delta w] =&- (\eta/\tau)e^{-\tilde{\gamma}_j\varepsilon}[\varphi(w_1)-(1-e^{\tilde{\gamma}_j\varepsilon} )\gamma_{i}\gamma_{j}  -\varepsilon \tilde{\gamma}_j].
\end{align*}
In the limit $\varepsilon \ll 1$, we get the desired approximation:
$$
\mathbb{E}[\Delta w] \approx -C_2  e^{-{\gamma}_j\varepsilon},\quad C_2= (\eta/\tau)\varphi(w_1) > 0,
$$
which is the negative branch of the STDP-like curves.
\end{proof}

\section{C: Training McCulloch-Pitts Networks}
\label{sec:experiments}

The pseudocode for MPN training is shown in Algorithm \ref{train}. Parameters are updated according to appropriate gradients of the transition terms and the holding terms. 

\begin{algorithm}[!hb]{}
\caption{MPN Training Algorithm}
\label{train}
\begin{algorithmic}[1]
\STATE $\mathcal{D}(t) \gets$ training data 
\STATE  Initialize $\theta = \{ \bm{w},\bm{b} \}$
\WHILE {not convergence of $\bm{w},\bm{b}$}
  \FOR {$n$ from $1$ to $N$}
    \STATE $\Delta_1 \bm{w}, \Delta_1 \bm{b} \gets $ gradients from $\mathcal{H}_n(\theta)$
    \STATE Update parameters $\bm{w}, \bm{b}$ using $\Delta_1 \bm{w}, \Delta_1 \bm{b} $
    \STATE $\Delta_2 \bm{w}, \Delta_2 \bm{b} \gets $ gradients from $\mathcal{T}_n(\theta)$
    \STATE Update parameters $\bm{w}, \bm{b}$ using $\Delta_2 \bm{w}, \Delta_2 \bm{b} $
  \ENDFOR
\ENDWHILE
\end{algorithmic}
\end{algorithm}

\begin{figure}[t]
  \centering
   {\includegraphics[scale=0.5]{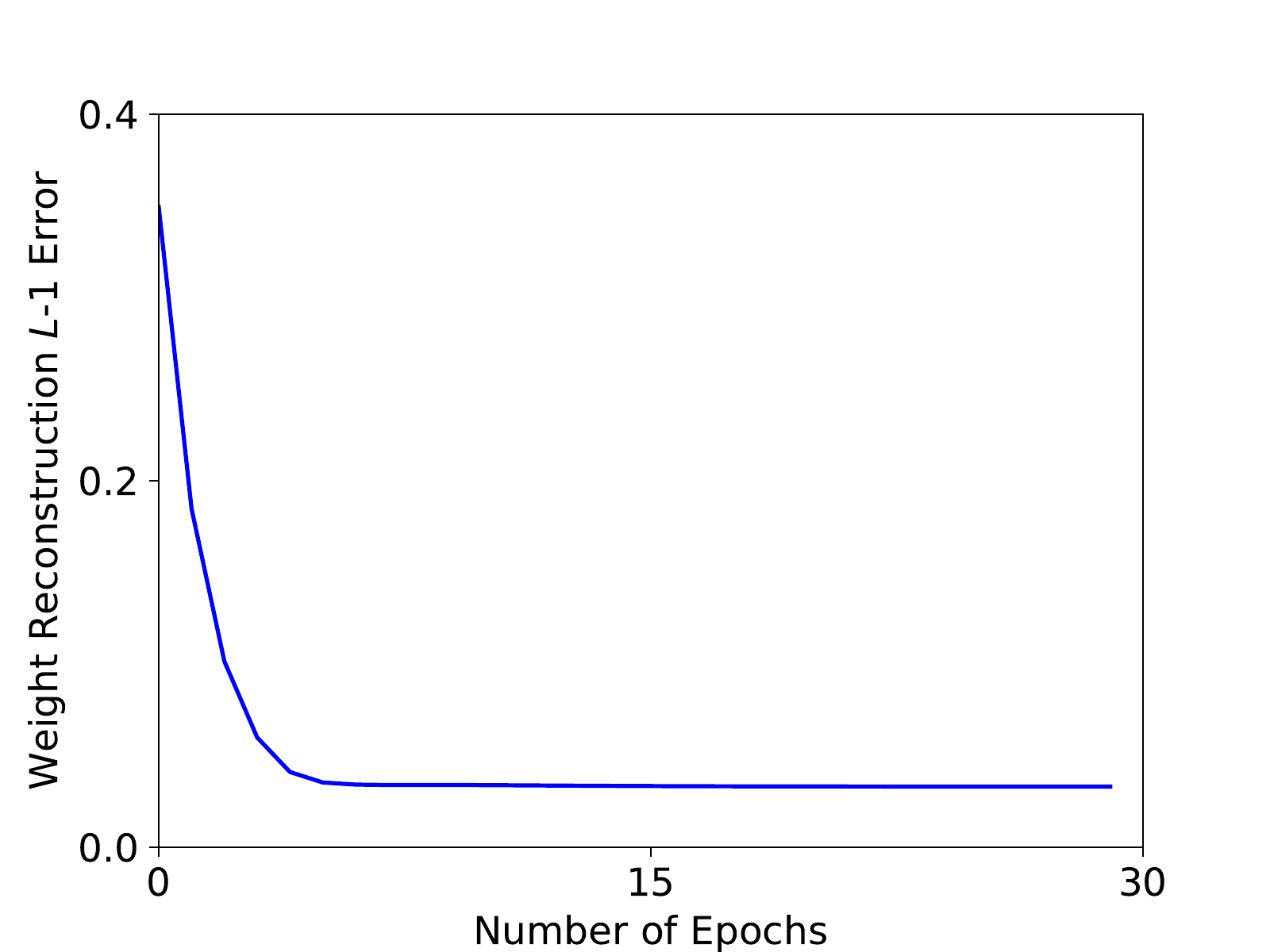}}
   \caption{ Weights reconstruction error during training. }
  \label{consistency2}
\end{figure}

\begin{figure}[t]
  \centering
  {\includegraphics[scale=0.27]{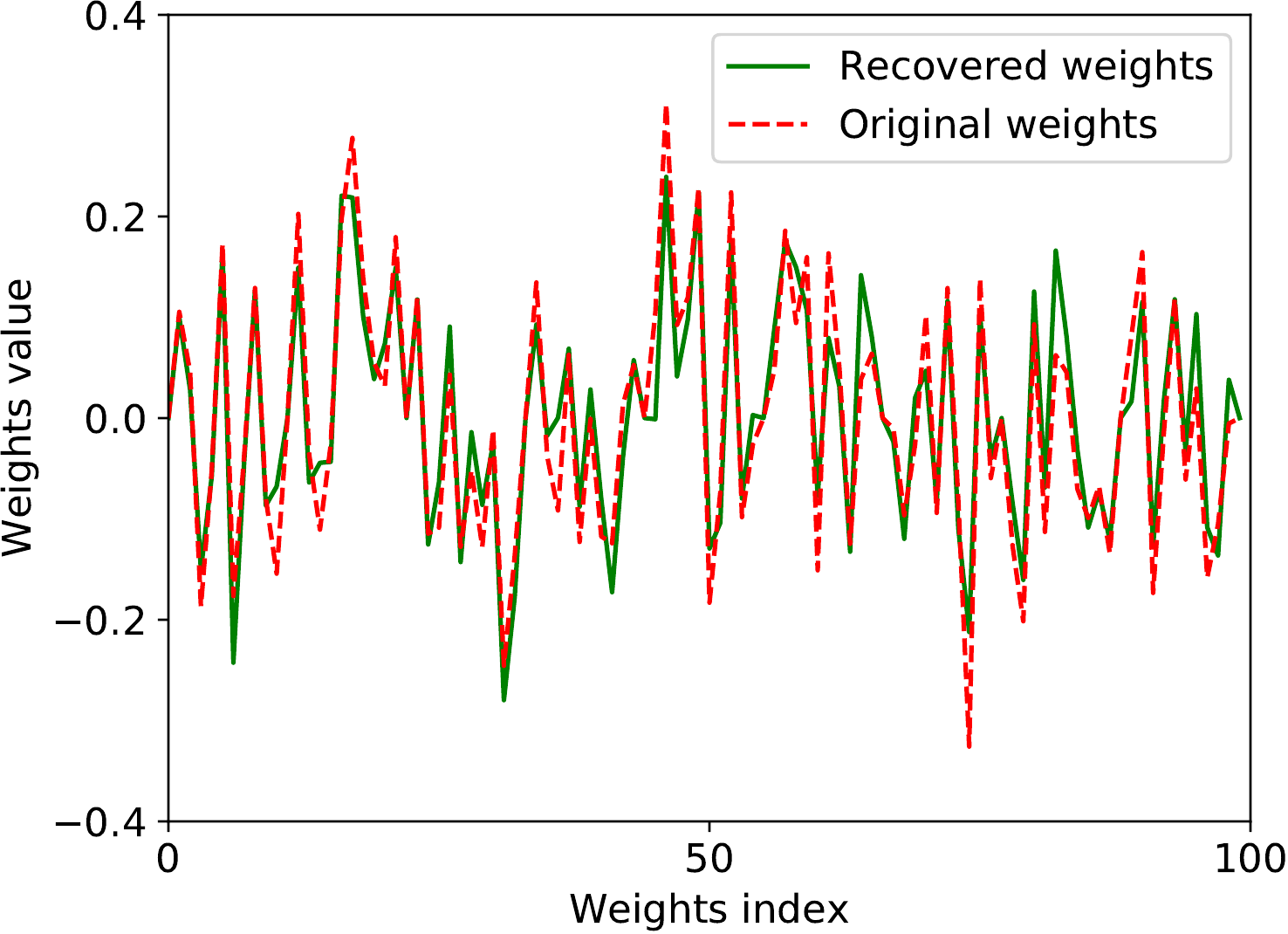}} 
   {\includegraphics[scale=0.27]{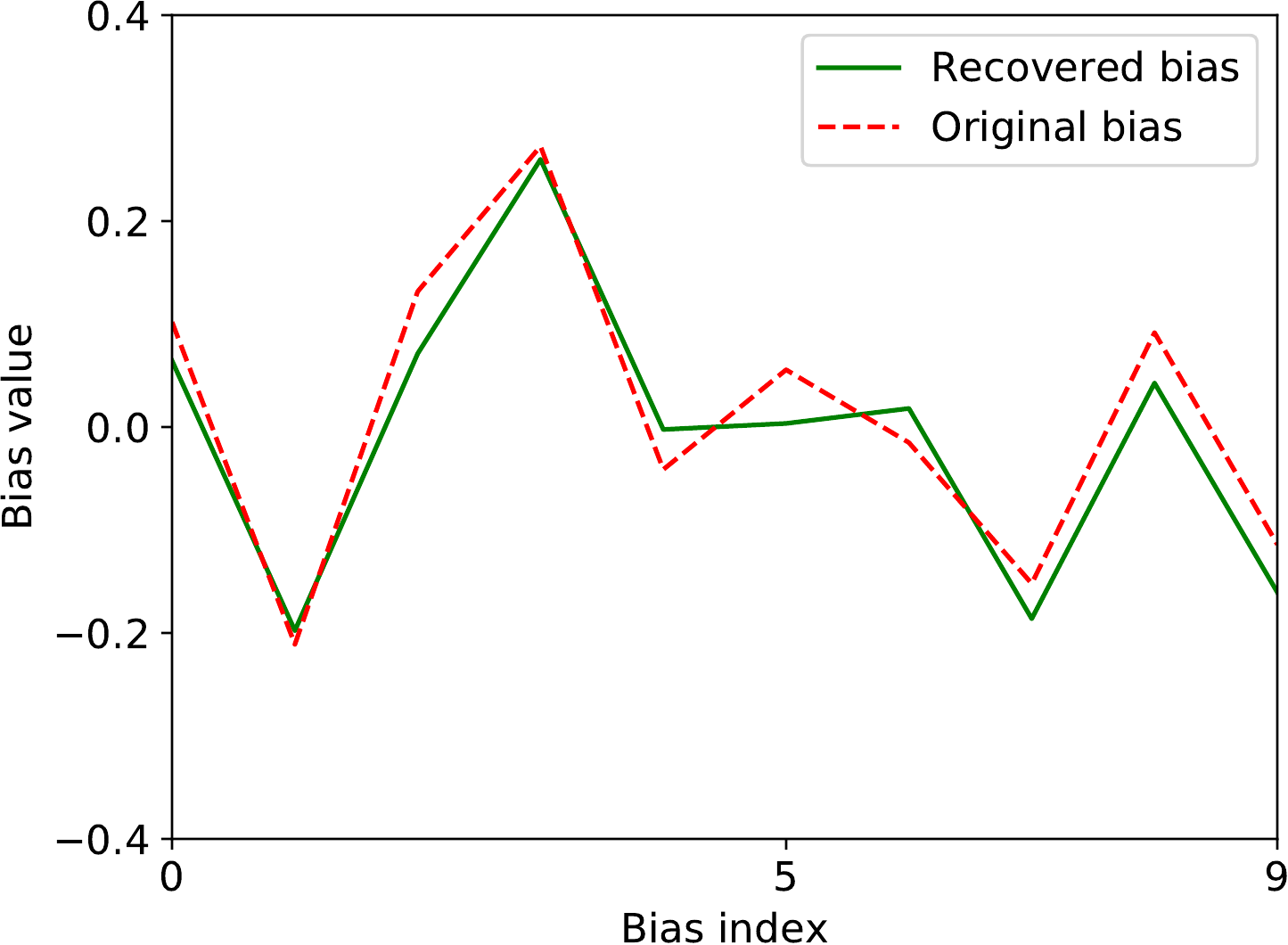}} 
   \caption{ Parameters inferred from the MPN learning rule agree with the parameters of the MPN used to generate the dynamics. This shows that the MPN is self-consistent. }
  \label{consistency}
\end{figure}

\section{D: Experimental Results} 
The codes for training and manipulating the McCulloch-Pitts network, including the videos demonstrating MPN's spatiotemporal pattern completion capability will be available at \url{https://github.com/owen94/MPNets}.
\subsection{1. Consistency}

The weight reconstruction error during training is shown in Fig.\ref{consistency2}. We notice that our algorithm converges  quite fast, reaching the minimum only after 5 epochs.
Fig.\ref{consistency} compares the recovered weights $\hat{w}$ and biases $\hat{b}$ to those of the original ones.

\subsection{2. Robust Spatiotemporal Memory}
In the main paper, we show how MPNs can learn two disjoint cycles. Here we visualize the topology of transitions between neuron states for the multiple cycles experiments in Fig.\ref{fig: 2cycles_topology}.

\begin{figure}[t]
\begin{center}
  {\includegraphics[scale=0.3]{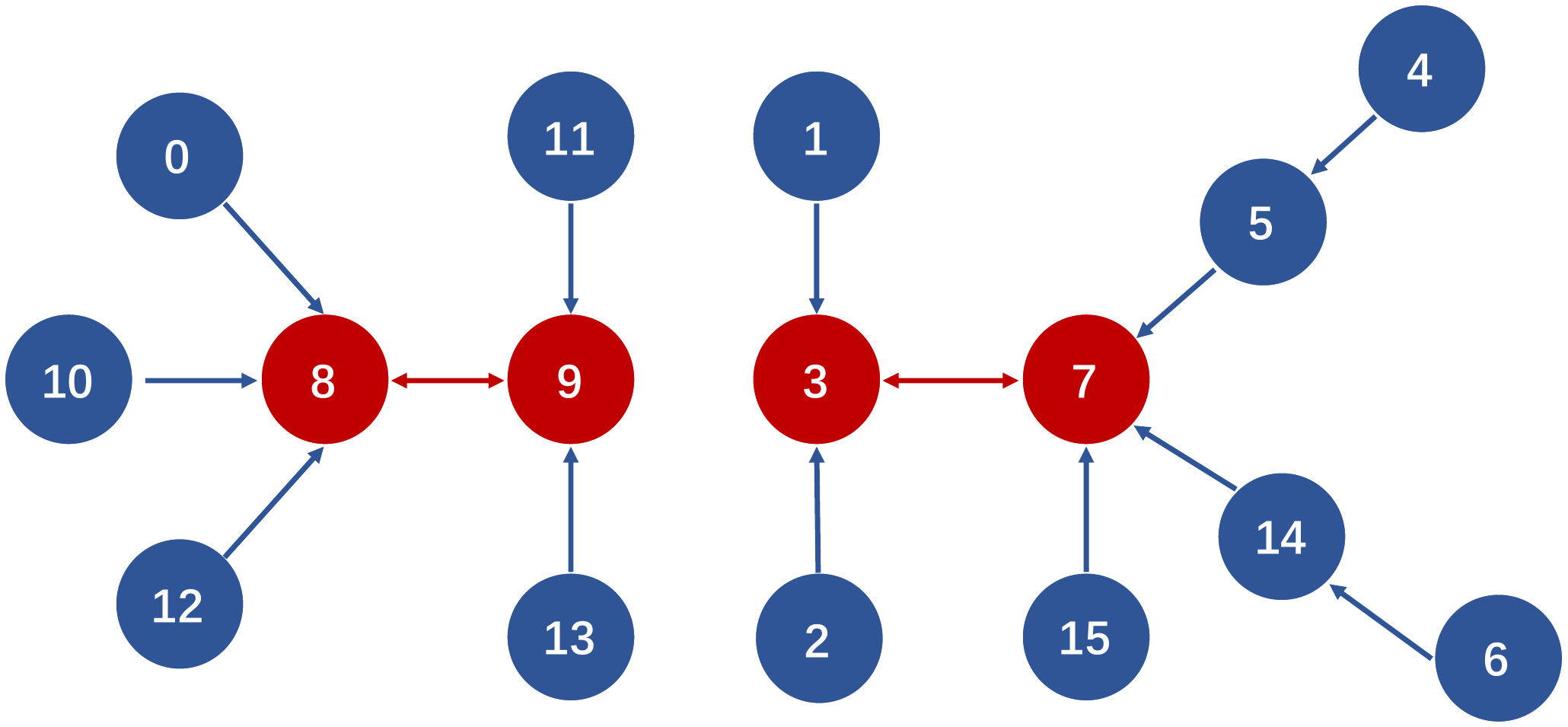}} 
  {\includegraphics[scale=0.3, center]{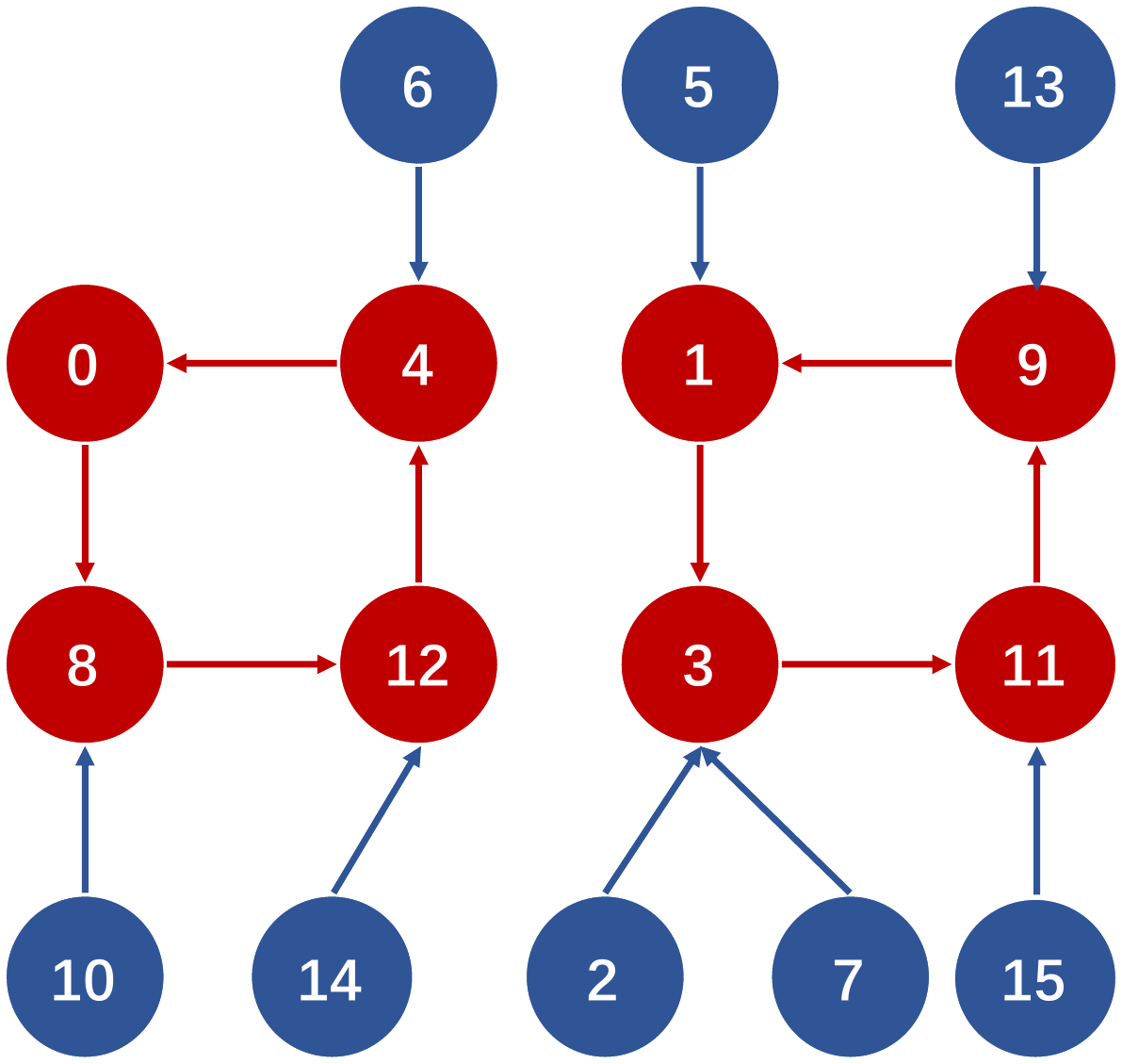}} 
    \caption{Deterministic state diagrams generated by an MPN with 4 neurons trained to store two robust cycles. The number on each node is the decimal representation of the corresponding state binary vector of length 4, e.g. the state $(1,0,1,0)$ is represented by $1 \times 2^0 + 0 \times 2^1 + 1 \times 2^2 + 0 \times 2^3 = 5.$ Blue nodes are transient states that flow along the arrows toward one of the two stable cycles. Red nodes are elements of stable cycles. Which robust cycle (stable periodic attractor) is reconstructed depends on the basin of attraction for the initial state. \textbf{(top)} shows the scenario when two memorized cycles are of period 2, while \textbf{(bottom)} contains two memorized cycles of period  $ 4$. In the presence of noise, transitions between the two deterministically disjoint cycles are possible, most likely via the minimum number of bits required to flip to reach the other cycle (measured by hamming distance). The minimal hamming distance between the two memorized cycles in \textbf{(top)} is 2 (e.g., The state can transverse via noisy dynamics from one cycle to the other via $\{(1,0,0,1) \rightarrow (0,0,0,1) \rightarrow (0,0,1,1)\} $, while the minimal hamming distance in \textbf{(bottom)} is 1. Hence, configurations \textbf{(bottom)} are more prone to cycle switching than \textbf{(top)} in the presence of noise, as illustrated in Fig.\ref{fig: 2cycles} in the main text. }
  \label{fig: 2cycles_topology}
\end{center}
\end{figure}

%
\begin{figure}[t]
  \centering
  {\includegraphics[scale=0.5]{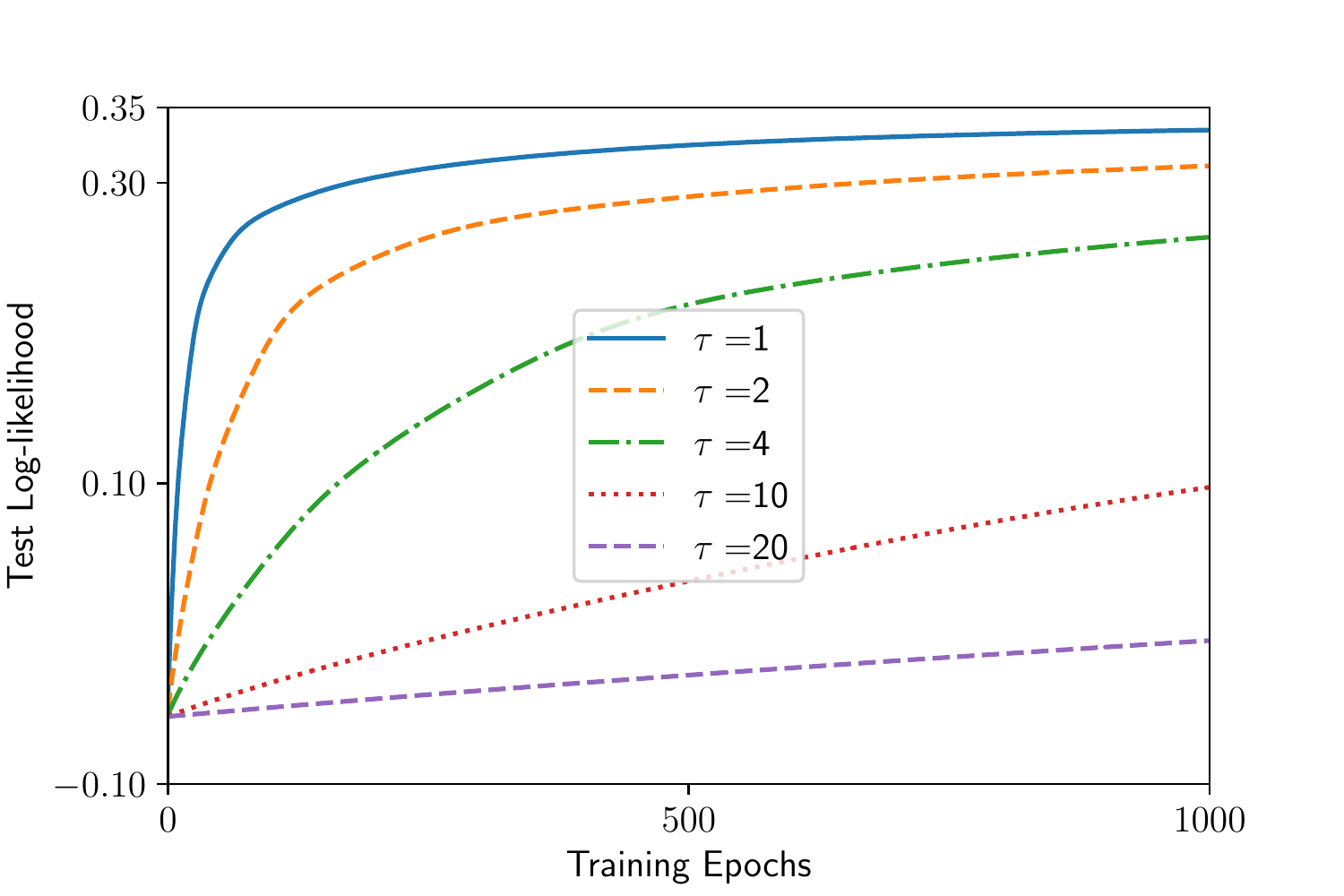}} 
   \caption{The test log-likelihood in Tim Blanche dataset.}
  \label{lld}
\end{figure}

\begin{figure*}[th!]
\label{fig: trainlld}
  \centering
  {\includegraphics[scale=0.5]{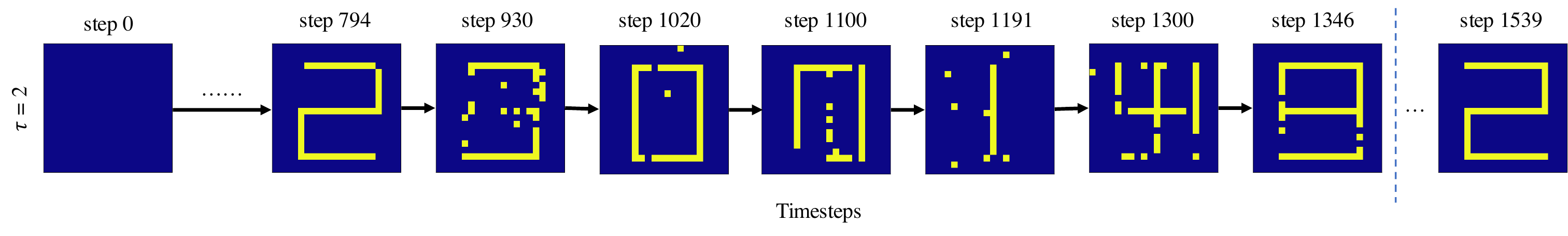}} 
   \caption{Spatiotemporal pattern completion: dynamics generated by the trained MPN with $\tau=2$ }
  \label{icmlseq}
\end{figure*}


\subsection{3. Application: Robustly Memorizing Sequences of Binary Pictures}

As explained in the main paper, we randomly assign 100 one-hop transitions between every successive pair of the static memories; i.e., the length for the entire training sequence, including the transitions from the all-zero state to ``2", ``2" to ``0", ``0" to ``1", and ``1" to ``9" is 400. Here, we show how we generate such training data by taking transitions from ``2" to ``0" as an example.

The model to be trained consists of $20 \times 20 = 400$ neurons that are originally all disconnected. Suppose the model begins with the spiked states ``2"; our goal is to bias the dynamics to transit into state ``0" through multiple one-hop transitions. To do this, we clamp large positive $z_i $ values for the neurons $i$ that are supposed to spike ($x_i = 1$) at the spiked-state of ``0", and large negative $z_j$ values for neurons $j$ that are supposed to be armed ($x_j=0$) at the armed-state of ``1". Then we generate the data following dynamics described in the subsection Dynamics of the main paper, where we use the predefined $z_i, z_j$ rather than using $z_i =\sum_{ji\in E} w_{ji}x_j + b_i$. 
Consequently, we can generate a sequence that stochastically transits from ``1" to ``0". Note that in this scenario, the dynamics is not guaranteed to exactly end up after 100 steps at ``0". Instead, it may end up at states nearby ``0" (nearby in Hamming distance). In our experiments, we set the  hyper-parameter $z_i = - z_j$ and choose  $z_i$ such that $\frac{z_i}{\tau} = 6$ in order to obtain transitions that can transit from one static memory to the other within 100 steps.  Timings are also sampled accordingly. 

From the procedure for generating training data mentioned above, we train the MPN based on Algorithm.\ref{train}. The dynamics generated from the trained MPN for $\tau=0.5$ and $\tau=1$ are shown in the main paper. For $\tau=2$, see Fig.\ref{icmlseq} below. 

\textbf{Videos:} We also created videos for the dynamics generated by the trained MPN under $\tau=0.5$, $\tau=1$, and $\tau=2$. The states (images) used for the videos are sampled as follows. Given the all-zero input state, we run the dynamics for 2300 one-hop transitions (to ensure a few ``2019" sequences will be repeated) and create a video with 32 frames per second (FPS) and 64 FPS. We can see that the trained MPNs can robustly recall the multi-hop away binary pictures in the correct sequential order, thus MPNs can perform robust spatiotemporal pattern completion! 

\subsection{4. Application: Inferring Generative Models for Neural Spike-train Data}

We preprocess the inter-spike intervals (ISI) data measured in microseconds($\mu s$) to be measured in seconds ($s$) to avoid numerical overflow. We set the refractory period to 200$\mu s$ for every spike, while the minimum experimental ISI is measured at 400$\mu s$. 

\begin{figure*}[t]
  {\includegraphics[width=1.\textwidth, center]{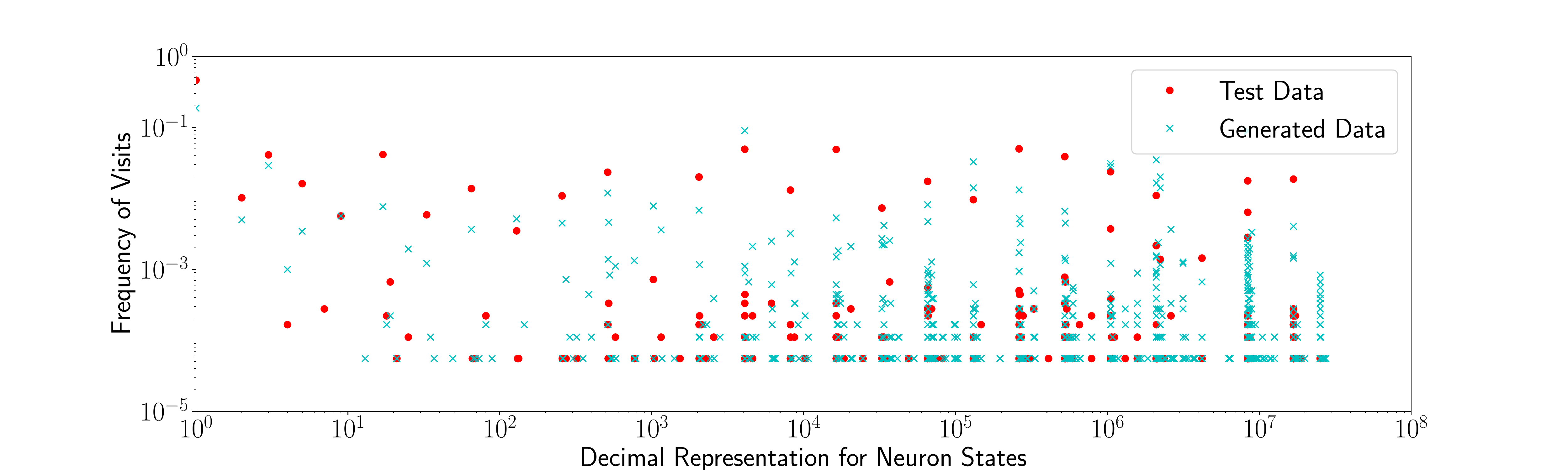}} 
   \caption{Histogram of the neuron states in test data and generated data. The trained MPN can reproduce the statistics of the highly visited neuron spike-states in the experimental spike-train data very well.}
     \label{hist}
\end{figure*}
\begin{figure*}[t]
  {\includegraphics[width=1.\textwidth, center]{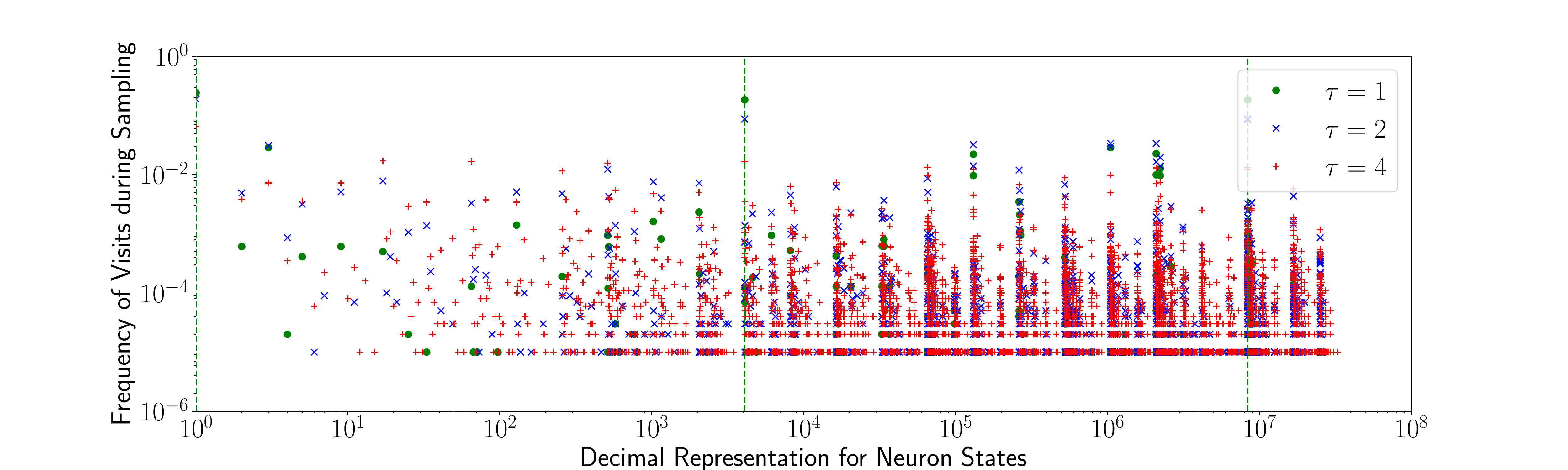}} 
   \caption{Histogram of neuron states collected from the dynamics generated by the MPN trained from Tim Blanche experiments at different $\tau$. The temperature $\tau$ plays a role of intrinsic noise that drives the dynamics away from the states visited in the (very low temperature) deterministic limit. More states will become accessible at higher temperature.}
  \label{taus}
\end{figure*}

The test log-likelihood of the MPNs trained at different $\tau$ are shown in Fig.\ref{lld}. 
There is no overfitting in the test log-likelihood, and it achieves the maximum of 0.33 when $\tau=1$. 
Notice that as $\tau$ increases, the final achieved log-likelihood will decrease and learning becomes slower. This phenomenon is consistent with theoretical results as the objective function and the gradients depend exponentially on $1/ \tau$, leading to smaller log-likelihood and slower learning. We expect
the learning to be faster for larger $\tau$ if we also appropriately increase the learning rate. 

We also plot the visited frequency of all the $2^{25}$ neuron states in Fig.\ref{hist}. The binary states are represented as their decimal representations, and the frequency of visits at a neuron state is computed from $\frac{\# \text{transits at such state}}{\text{\# of all transits}}$. Notice that the statistics of these highly visited states in the test and in the generated data agree quite well. 

We also explore the dynamics generated by the trained MPNs at different temperature $\tau$, shown in Fig.\ref{taus}. As expected, one can clearly see that the temperature $\tau$ plays a role of intrinsic noise that drives the dynamics away from the states visited in the (very low temperature) deterministic limit. More states will be accessible at higher temperature.

%

\end{appendix}

\end{document}